%% file: paper.tex
\documentclass[12pt]{article}
\usepackage{fullpage}
\usepackage{amssymb}
\usepackage{amsmath}
\usepackage{amsthm}

\usepackage{nicefrac}
\usepackage{url}
\usepackage{color}
\usepackage[dvipsnames]{xcolor}
\usepackage[normalem]{ulem}

\usepackage{graphicx}
\usepackage{subcaption}
\usepackage[textsize=scriptsize,disable]{todonotes}
\usepackage{xspace}
\usepackage{microtype}  
\usepackage{hyperref}
\usepackage[noabbrev,capitalise]{cleveref}  
\usepackage{xifthen}

\newtheorem{theorem}{Theorem}
\newtheorem{corollary}{Corollary}
\newtheorem{definition}{Definition}

\newtheorem{claim}{Claim}
\newtheorem{observation}{Observation}

\newtheorem{lemma}{Lemma}

\crefname{enumi}{Step}{Steps}

\newcommand{\rpw}{\ensuremath{{\mathit rpw}}}

\newcommand{\Sd}{\ensuremath{{\mathit H\!S}}}
\newcommand{\pw}{\ensuremath{{\mathit pw}}}
\newcommand{\W}{{\text{\rm W}\xspace}}
\newcommand{\E}{\text{\rm E}\xspace}
\newcommand{\N}{\text{\rm N}\xspace}
\renewcommand{\S}{\text{\rm S}\xspace}
\newcommand{\LLL}{\text{\rm WWW}\xspace}
\newcommand{\LUL}{\text{\rm WNW}\xspace}
\newcommand{\LDL}{\text{\rm WSW}\xspace}
\newcommand{\LLR}{\text{\rm WWE}\xspace}
\newcommand{\RLL}{\text{\rm EWW}\xspace}
\newcommand{\RLR}{\text{\rm EWE}\xspace}

\newcommand{\calI}{{\ensuremath{\cal I}}}
\newcommand{\calC}{{\ensuremath{\cal C}}}
\newcommand{\calT}{{\ensuremath{\cal T}}}

\newcommand{\leaveout}[1]{}
\newcommand{\student}[1]{}
\newcommand{\postdoc}[1]{}

\newcommand{\therese}[2][]{\ifthenelse{\isempty{#1}}
{\todo[color=Green]{#2}}
{\todo[#1,color=Green]{#2}}}
\newcommand{\milap}[2][]{\ifthenelse{\isempty{#1}}
{\todo[color=Cyan]{#2}}
{\todo[#1,color=Cyan]{#2}}}

\renewcommand{\medskip}{\smallskip}
\renewcommand{\int}{{\ensuremath{\rm int\,}}}

\date{}
\title{Drawing Halin-graphs with small height}
\author{Therese Biedl%
\thanks{David R.~Cheriton School of Computer
Science, University of Waterloo, Waterloo, Ontario N2L 1A2, Canada.
TB supported by NSERC.  {\tt \{biedl,m2sheth\}@uwaterloo.ca}.  }
\and
Milap Sheth\addtocounter{footnote}{-1}\footnotemark
}

\begin{document}

\maketitle
\begin{abstract}
In this paper, we study how to draw Halin-graphs, i.e., planar graphs
that consist of a tree $T$ and a cycle among the leaves of that tree.
Based on tree-drawing algorithms and the pathwidth $ \pw(T) $, a well-known graph parameter,
we find poly-line drawings of height at most $6\pw(T)+3\in O(\log n)$.
We also give an algorithm for straight-line drawings, and
achieve height at most $12\pw(T)+1$ for Halin-graphs, and
smaller if the Halin-graph is cubic.
We show that the height achieved by our algorithms is optimal in the worst case (i.e. for some Halin-graphs).
\end{abstract}

\section{Introduction}

It is well-known that every planar graph has a planar straight-line drawing
in an $O(n)\times O(n)$-grid \cite{FPP90,Sch90} and that an $\Omega(n)\times \Omega(n)$-grid
is required for some planar graphs \cite{FPP88} (definitions will be given in
the following section).  But for some subclasses of planar graphs, planar
straight-line drawings of smaller area can be found.  In particular, for any
tree one can easily create a straight-line drawing of area $O(n\log n)$ \cite{CDP92};
the area can be improved to $n2^{O(\sqrt{\log\log n \log\log \log n})}$ \cite{Chan18}
and $O(n)$ if the maximum degree is $O(n^{1-\varepsilon})$ \cite{GGT96}.
Outerplanar graphs can be drawn with area $O(n^{1.48})$ \cite{DBF09}
and with area $O(n\log n)$ if the maximum degree is bounded
\cite{Frati12} or a constant number of bends are allowed in edges \cite{Bie-DCG11}.
There are also some sub-quadratic area results for series-parallel graphs \cite{Bie-DCG11}, though they require bends in the edges.

These existing results suggest that bounding the so-called {\em treewidth} of a graph may
be helpful for obtaining better area bounds.  In particular, trees have treewidth 1, and
outer-planar and series-parallel graphs have treewidth 2.  However, one can observe that the
lower-bound graph from \cite{FPP88} can be modified to have treewidth 3, so we cannot hope
to achieve subquadratic area for all planar graphs of constant treewidth.  However, there
are some subclasses of planar graphs that have treewidth 3 and a special structure that may
make them amenable to be drawn with smaller area.  This is the topic of the current paper.

Halin-graphs were originally introduced by Halin \cite{Halin71} during his study
of graphs that are planar and 3-connected and minimal with this property.  He showed that
any such graph consists of a tree without vertices of degree 2
where a cycle has been added among the leaves of the tree.  These graphs have attracted
further interest in the literature, see for example \cite{SP83,SS90,FT06,FL16,Epp16}.
It is folklore that they can be recognized in linear time since they are planar graphs
and have treewidth 3, but a direct and simpler approach for this was recently given
by Eppstein \cite{Epp16}.

In this paper, we study how to create planar drawings of a Halin-graph that have small area.
To our knowledge, no such algorithms have been given before, and the
best previous result is to apply a general-purpose planar graph drawing algorithm that
achieves area $O(n^2)$.  In contrast to this, we exploit here that a Halin-graph consists
of a tree $T$ with a cycle $C$ among its leaves, and give two results.  The first one states that
for any drawing of $T$, we can ``fiddle in'' the cycle $C$ at a cost of increasing the height
by a factor of 3.  However, the resulting drawing has bends.  For our second result, we
take inspiration from one particular tree-drawing algorithm by Garg and Rusu \cite{GR03} to
create an algorithm that achieves straight-line drawings of area $O(n\log n)$.
In fact, the height of our drawings, which is $O(\log n)$ in the worst case, can be bounded more
tightly by $O(\pw(T))$, where the {\em pathwidth} $\pw(T)$ is a well-known graph-parameter.
It is known that the pathwidth is a lower bound on the height of any planar graph drawing \cite{FLW03}
and that the pathwidth of a Halin-graph is within a constant factor of the pathwidth of the tree $T$
\cite{FT06}.  Therefore our algorithm gives a $O(1)$-approximation algorithm on the height
of planar straight-line drawings of Halin-graphs if we ignore small constant terms.
Similarly as
was done for trees by Suderman \cite{Sud04} and Biedl and Batzill \cite{BB-JGAA}, we can also
argue that the constant in front of ``$\pw(T)$'' cannot be improved for some Halin-graphs.

Our paper is structured as follows.  After reviewing the necessary background in
Section~\ref{sec:definitions}, we briefly argue in Section~\ref{sec:transform}
how to use any tree-drawing
algorithm to create (poly-line) drawings of Halin-graphs of asymptotically
the same height.  Section~\ref{sec:straight} gives the algorithm for
straight-line drawings of small height, while Section~\ref{sec:lower} defines
a class of Halin-graphs that have small pathwidth, yet require a large height
in any (straight-line or poly-line) planar drawing.  We conclude in Section~\ref{sec:conclusion}.

\section{Background and notations}
\label{sec:definitions}

We assume familiarity with graphs and basic graph-theoretic terms, see for example \cite{Die12}.
Throughout this paper, we use $n$ for the number of vertices in a given graph $G=(V,E)$.
A {\em tree} is a connected graph without cycles.  A {\em leaf} of a tree is a vertex
of degree 1.  A {\em rooted tree} is a tree together
with one specified vertex (the {\em root}); this defines for any edge of the tree
the parent-child relationship with the {\em parent} being the endpoint that is closer
to the root.  In a rooted tree, the term {\em leaf} is used only for those vertices
that have no children, i.e., the root is not considered a leaf unless $n=1$.

Fix a rooted tree $T$. For any vertex $v\in T$, we use $T_v$ to denote
the {\em subtree of $T$ rooted at $V$}, i.e., vertex $v$ and all its
descendants.
We assume throughout that trees are {\em ordered}, i.e., come with a
fixed cyclic order of neighbours around each vertex.  In a rooted tree,
this hence gives a left-to-right order of its children (starting in counter-clockwise
direction after the parent).  The {\em leftmost} leaf $\ell_L$ of $T$ is
the one reached by starting at the root and repeatedly taking the
leftmost child until we reach a leaf.  Define the {\em rightmost leaf}
$\ell_R$ symmetrically.  Note that $\ell_L=\ell_R$ if $T$ is a {\em rooted path},
i.e., it is a path with the root as one of its endpoints.    If $T$ consists
of only one vertex (the root $r$), then $\ell_L=r=\ell_R$, but otherwise
$\ell_L\neq r\neq \ell_R$.

\paragraph{Halin-graphs and skirted graphs:}
Let $T$ be an (unrooted, ordered) tree without vertices of degree 2.
To avoid trivialities, we assume that $T$ has at least three leaves.
Let $H$ be the graph obtained by connecting the leaves of $T$
in cyclical order; this is the {\em Halin-graph} formed by $T$ (and sometimes
denoted $H(T)$). Tree $T$ is called the {\em skeleton}
of Halin-graph $H$, and the edges of the cycle are called {\em cycle-edges}.
See Figure~\ref{fig:HalinExample}.

Observe that any Halin-graph is {\em planar}, i.e.,
can be drawn without crossing in the plane.
The condition `no vertex has degree 2' is not crucial for our drawing
algorithm (though it was crucial in the original study of Halin-graphs
as minimal 3-connected planar graphs \cite{Halin71}).  As in \cite{FT06},
we use the term {\em extended Halin-graph} for a graph $H(T)$ obtained by
taking an arbitrary  tree $T$ and connecting its leaves in a cycle in order,
while a {\em regular Halin-graph} refers to a Halin-graph as above, i.e.,
the skeleton has no vertices of degree 2.

\begin{figure}[ht]
  \centering
\begin{subfigure}[b]{0.4\linewidth}
\includegraphics[width=0.99\linewidth,page=1,trim=0 0 0 0,clip]{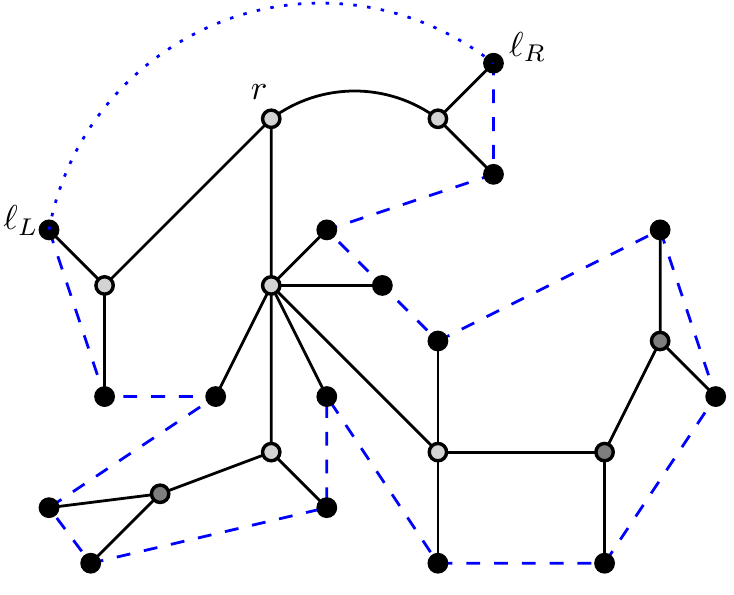}
\caption{}
\end{subfigure}
\hspace*{\fill}
\begin{subfigure}[b]{0.45\linewidth}
\includegraphics[width=0.99\linewidth,page=7,trim=0 0 0 0,clip]{transform3.pdf}
\caption{}
\label{fig:transform_final}
\end{subfigure}
\caption{(a) A regular Halin-graph. Cycle-edges are blue dashed/dotted, skeleton $T$ is black/gray, and the skirted graph $H^-(T)$
would omit the dotted edge if $T$ were rooted at $r$.  The inner skeleton is gray, the leaf-reduced inner skeleton
is light gray.  (b) A poly-line drawing obtained with the transformation in
Section~\ref{sec:transform}.}
\label{fig:HalinExample}
\end{figure}

Our drawing algorithms will be based on rooted, rather than unrooted, trees,
and therefore exploit subgraphs of Halin-graphs formed by rooted trees.
Let $T$ be an (ordered) tree that has been rooted at vertex $r$.
Let $H$ be the graph obtained by connecting the leaves of $T$ in order from
left to right in a path; this is the {\em skirted graph} \cite{SS90}
formed by $T$ (and sometimes denoted $H^-(T)$).
Graph $H^-(T)$ is a subgraph of $H(T)$; it is missing either the
edge $(\ell_L,\ell_R)$ or (if the root $r$ has degree 1) the path $\langle \ell_L,r,\ell_R\rangle$.

\paragraph{Pathwidth and rooted pathwidth: }
The {\em pathwidth} of a graph $G$ is defined as follows.  A {\em path
decomposition} is an ordered sequence $X_1,\dots,X_\xi$ of vertex-sets ({\em bags})
such that any vertex belongs to a non-empty subsequence of bags, and for any edge at
least one bag contains both endpoints.  The {\em width} of such a path
decomposition is $\max_i \{|X_i|-1\}$, and the pathwidth $\pw(G)$ is the
minimum width of a path decomposition of $G$.  A graph consisting of a
singleton vertex hence has pathwidth 0.

We will in this paper almost only be
concerned with the pathwidth of trees; here an equivalent and simpler definition
is known.  For a path $P$ in a tree $T$, let $\calT(T,P)$ denote the connected
components of the graph obtained by removing the vertices of $P$.  Suderman
\cite{Sud04} showed that for any tree $T$ we have
\[
  \pw(T) :=  \begin{cases}
    0 & \text{if $T$ is a single vertex,} \\
    \min_{P} \max_{T'\in \calT(T,P)}\{1+\rpw(T')\} & \text{otherwise,}
  \end{cases}
\]
where the minimum is taken over all paths $P$ in $T$.
%
%
Our constructions will use a rooted tree $T$, and therefore consider width-parameters
for rooted trees that are illustrated in Figure~\ref{fig:RpwExample}.
Define as in \cite{OPTII} the {\em rooted pathwidth} $\rpw(T)$ as follows:
\[
  \rpw(T) :=  \begin{cases}
    1 & \text{if $T$ is a rooted path,} \\
    \min_{P_r} \max_{T'\in \calT(T,P_r)}\{1+\rpw(T')\} & \text{otherwise,}
  \end{cases}
\]
where the minimum is over all rooted paths $P_r$ of $T$.  (The recursive
formula differs from the one for pathwidth only in that the path must end at
the root; hence the name.)
One can show that any tree $T$ can be rooted at
a leaf such that we have $\pw(T)\leq \rpw(T)\leq 2\pw(T)+1$ \cite{OPTII}.
We call a path $P_r$ that can be used to obtain the minimum a {\em spine}.

\begin{figure}[ht]
  \centering
\includegraphics[width=0.4\linewidth,page=2,trim=0 0 0 10,clip]{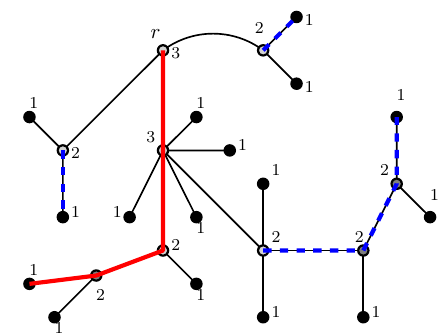}
\caption{The skeleton-tree $T$ of Figure~\ref{fig:HalinExample}
has $ \Sd(T) = \rpw(T) = 3 $. Numbers indicate the Horton-Strahler number; thick paths
(solid red for the whole tree, dashed blue for the subtrees) are possible spines.}
\label{fig:RpwExample}
\end{figure}

The rooted pathwidth was actually used much earlier for the classification of the order
of rivers and streams \cite{Horton45, Strahler52} and became known as the Horton-Strahler number:
\[
  \Sd(T) :=  \begin{cases}
    1 & \text{if $T$ is a single vertex,} \\
    \min_{c} \max_{v} \{ \Sd(T_v)+\chi(v\neq c)\} & \text{otherwise,}
  \end{cases}
\]
where the minimum is over all children $c$ of the root $r$, the maximum $v$ is over
all children of the root, and $\chi$ denotes the characteristic function.
One can show \cite{OPTII}
that the Horton-Strahler number and the rooted pathwidth are identical.
We use the term {\em spine-child} for a child $c$ where the minimum is achieved;
this is the same as a child that maximizes the Horton-Strahler number among the children.
(One can show that it belongs to a spine of $T$.)

\paragraph{Graph drawing:}  A {\em poly-line} is a polygonal curve, i.e., a
curve that is the union of finitely many line segments; the transition between
two such segments is called a {\em bend}.
A {\em planar poly-line drawing} $\Gamma$ of a graph $G$
consists of assigning a point to each vertex and an (open) poly-line to each edge
such that all points and poly-lines are disjoint, and the poly-line of an edge ends at the
points of the endpoints of the edge.
The drawing is called {\em $y$-monotone} if all
poly-lines  of edges are $y$-monotone and {\em straight-line} if all poly-lines of edges are straight-line segments.

We assume throughout that identifying features (i.e., points of vertices and bends in poly-lines of edges)
have integral $y$-coordinates.  The {\em layers} of a drawing are the horizontal
lines with integral $y$-coordinate that intersect the drawing; we usually enumerate
them from top to bottom as $1,2,\dots,h$.  The number $h$ of layers is called the
{\em height} of the drawing (notice that this is one unit more than the height
of the minimum enclosing box).    Minimizing the height of drawings is the main
objective in this paper.  When constructing drawings, it will sometimes be
expedient to use integral $x$-coordinates as well; we then use the term {\em column}
for a vertical line of integral $x$-coordinate that intersects the drawing and
enumerate columns from left to right.

We usually identify the graph-theoretic
object (vertex, edge) with the geometric object (point, poly-line) that corresponds
to it in the drawing.
All our drawings are required to be {\em planar} (i.e., without crossing edges)
by definition.  We often require that they are {\em plane}, i.e. reflect the
given order of edges around every vertex, and (for a Halin-graph) the infinite
region is adjacent to the cycle-edges.

\section{Transforming tree drawings}
\label{sec:transform}

\input{transform2.tex}

\section{Straight-line drawings}
\label{sec:straight}

\input{straight4.tex}

\section{Lower bounds on the height}
\label{sec:lower}

\input{lower2.tex}

\section{Conclusion}
\label{sec:conclusion}

In this paper, we studied drawings of Halin-graphs whose height is
within a constant factor of the optimum (ignoring small additive terms).  We gave a 6-approximation
for the height of poly-line drawings of such graphs, and a 12-approximation
for the height of straight-line drawings.  We also showed that there
exists a Halin-graph for which our constructions give the minimum
possible height.  Many open problems remain:
\begin{itemize}

\item Can we find straight-line drawings of height $c\cdot \pw(T)+O(1)$,
	for $c<12$ and ideally $c=6$?

\item We have focused on the height and ignored the width.
	For straight-line drawings, the detour through $y$-monotone
	drawings means that the width may be exponential.  Are there
	straight-line drawings of height $O(\pw(T))$ for which the
	width is polynomial (and preferably linear)?
\end{itemize}

Last but not least, are there other planar graph classes that have approximation
	algorithms for height (or perhaps the area) of planar graphs drawings?

\bibliographystyle{plain}
\bibliography{full,gd,papers,refs}

\end{document}

%% file: transform2.tex
In this section, we show that any order-preserving tree-drawing algorithm
can be used to obtain poly-line drawings of Halin-graphs.  Put differently,
we can draw the skeleton-tree $T$, and ``fiddle in'' the cycle-edges.  As
it will turn out, we do not need to use a drawing of $T$; it suffices to
take a drawing of a suitably chosen subtree of $T$, which may make the
height bound a bit smaller and (as we will see) give a tight bound.

The following defines the subtree of $T$ that we draw; see also
Figure~\ref{fig:HalinExample}.
Let the {\em inner skeleton} of a Halin-graph
be the tree $T'$ obtained by deleting all leaves of the skeleton.
We say that $T'$ {\em leaf-extends} a tree $T''$
if $T'$ can be obtained from $T''$ by (possibly repeatedly) adding a leaf
incident to a leaf of the previous tree.  The {\em leaf-reduced inner skeleton} of
a Halin-graph $H(T)$ is the smallest subgraph of the inner skeleton $T'$ that can
be leaf-extended to $T'$.  We now have the following result:

\begin{theorem}
  \label{thm:3x-transform}
  Let $H(T)$ be an extended Halin-graph.  If its leaf-reduced inner skeleton $T''$
  has an order-preserving poly-line drawing $\Gamma''$ of height $h$,
  then $H(T)$ has a plane poly-line drawing of height $3h$.
\end{theorem}
\begin{proof}
Figure~\ref{fig:convert} illustrates how to find this drawing, with the
final result in Figure~\ref{fig:transform_final}.
As a first step, insert a dummy-vertex at every bend of $\Gamma''$
to get a straight-line drawing $\Gamma_d''$ of a tree $T''_d$ that is tree $T''$
with some edges subdivided.  Also subdivide the same edges in trees $T'$ and $T$
(where $T'$ is the inner skeleton of $H(T)$) to get trees $T'_d$ and $T_d$.

\begin{figure}[ht]
\begin{subfigure}[b]{0.15\linewidth}
\includegraphics[scale=0.59,page=3,trim=0 0 60 0,clip]{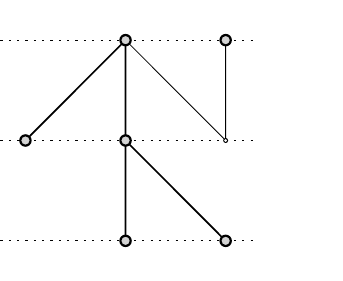}
\caption{}
\end{subfigure}
\hspace*{\fill}
\begin{subfigure}[b]{0.28\linewidth}
\includegraphics[scale=0.59,page=4,trim=0 0 0 0,clip]{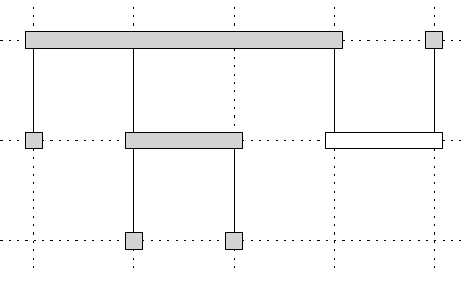}
\caption{}
\end{subfigure}
\hspace*{\fill}
\begin{subfigure}[b]{0.23\linewidth}
\includegraphics[scale=0.59,page=5,trim=0 0 20 0,clip]{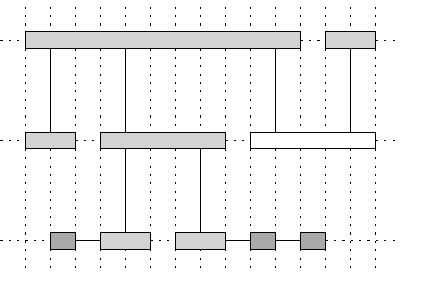}
\caption{}
\end{subfigure}
\hspace*{\fill}
\begin{subfigure}[b]{0.23\linewidth}
\includegraphics[scale=0.59,page=6,trim=0 0 20 0,clip]{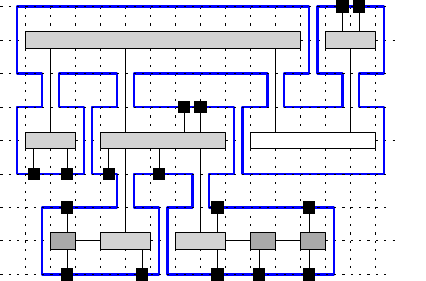}
\caption{}
\end{subfigure}
\caption{Transform (a) a poly-line drawing of the leaf-reduced inner skeleton
(with a white dummy-vertex inserted at a bend) into (b) a visibility represention.
(c) Expand leaves and widen vertex-segments to overhang ($x$-coordinates are not to scale).
Then (d) triple the grid and insert cycle $C$ and the leaves to get a flat orthogonal drawing (inserted columns are not shown).
}
\label{fig:convert}
\end{figure}

Next, convert $\Gamma_d''$ into a {\em flat
visibility representation} $\Gamma''_{vr}$ of $T_d''$.
This consists of assigning a horizontal segment $s(v)$ to every vertex
and a horizontal or vertical segment $s(e)$ to every edge such
that the segments are interior-disjoint and the segment of edge
$(v,w)$ ends at $s(v)$ and $s(w)$.    We can always do such a
conversion while giving integral $y$-coordinates to all segments
and maintaining the same height and planar embedding \cite{Bie-GD14}.

We next convert visibility representation $\Gamma''_{vr}$ of $T''_d$
into a visibility representation $\Gamma'_{vr}$ of $T'_d$.  Recall that
$T'$ is a leaf-extension of $T''$, so we can obtain $T'_d$ by repeatedly
adding a leaf $\ell$ incident to a leaf $p$ of the current tree.
Since $p$ is a leaf,
there is no incident horizontal edge next to one end (left or right) of its segment
$s(p)$. We place
a segment for $\ell$ at this end (inserting columns if needed to make space), and
connect it horizontally.  Repeating this gives a
visibility representation $\Gamma'_{vr}$ of $T'_d$.  By inserting
further columns, we may assume that any segment $s(v)$ in $\Gamma'_{vr}$
has at least one unit width and overhangs any incident vertical edge-segment by at least one unit.

Next, {\em triple the grid}, i.e.,  insert a new grid-line before and after each
existing one.
In consequence, we can surround the entire drawing of $\Gamma'_{vr}$ with
a cycle $C$ that traces along all segments.  Formally, $C$ consists of
all those points that are horizontally or vertically exactly one unit
away from segments of $\Gamma'_{vr}$, and these points form a cycle since
we tripled the grid.  Let $\Gamma'_C$ be
the resulting drawing.

Now we insert the leaves of the skeleton.
Let an {\em angle} of a vertex $v$ in $T'$ be any two consecutive
edges $e,e'$ at $v$ in $T'$ in the planar embedding.
Because $s(v)$ overhangs its incident vertical edges,
cycle $C$ has a segment $s_\alpha$ of at least unit length for every angle
$\alpha$ of $v$ such that placing a leaf $\ell$ on $s_C$ and connecting it vertically
puts edge $(v,\ell)$ between $e$ and $e'$ in the planar embedding.
So for any $v\in T'$ and any angle $\alpha$ at $v$, insert
Note that $C$ runs within unit distance of $s(e)$ and $s(e')$ at some point,
and since $e,e'$ are consecutive  at $v$, a part of $C$ between this is
within unit distance of $s(v)$ throughout.  Furthermore, since
$s(v)$ overhangs incident edge-segments, this part contains a horizontal
segment $s_\alpha$. Insert as many leaves
on $s_\alpha$ as are required by the planar embedding of the skeleton
(we can insert columns to widen $s_\alpha$ if needed) and connect
them vertically to $s(v)$.
This gives a {\em flat orthogonal drawing} $\Gamma_{od}$: every vertex is represented
by a horizontal segment, and every edge is a poly-line with only
horizontal or vertical segments.  Furthermore, the height is $3h$
and the drawing represents $H(T_d)$ since we took care to re-insert
the leaves exactly according to the planar embedding.
Drawing $\Gamma_{od}$ can be converted to a poly-line drawing $\Gamma_d$ of $H(T_d)$
of the same height \cite{Bie-GD14}.  Finally by reverting dummy-vertices
of $T_d$ back to bends, we obtain the desired poly-line drawing of $H(T)$.
\end{proof}

Since every tree $T$ has an order-preserving
straight-line drawing $\Gamma$ of height $2 \pw(T) + 1$
\cite{BB-JGAA}, we get:

\begin{corollary}
  \label{cor:6pw}
  Any extended Halin-graph $H(T)$ has a plane
  poly-line drawing of height $6 \pw(T'') + 3$, where $T''$ is the reduced inner skeleton.
\end{corollary}

Since every tree has pathwidth at most $\log_3(2n+1)$ \cite{Scheffler90}
we can in particular draw extended Halin-graphs with height $O(\log n)$.  The width can easily be seen to be $O(n)$, so the area is $O(n\log n)$.
Our construction may seem very wasteful (cycle $C$ has many bends
that could be removed with suitable post-processing stages),
but as we shall see in Theorem~\ref{thm:lower_regular_pw},
the height-bound is tight, even for some regular Halin-graphs.

%% file: straight4.tex
The transformation of Section~\ref{sec:transform} creates
poly-line drawings, and it is not at all clear
whether one could convert them into straight-line drawings without
changing the height.  We hence give a second, completely different algorithm
that creates a straight-line plane drawing of a Halin-graph that, at the
cost of doubling the height.  (The width may
be exponential, so this construction is of mostly theoretical interest.)
Crucial for our result is that it suffices to construct poly-drawings in
which all edges are drawn as $y$-monotone curves; by the result of
Pach and T\'{o}th  \cite{PT04} or Eades et al.~\cite{EFLN06} such drawings
can be converted into planar drawings of the same height.

The algorithm proceeds by considering an increasingly larger subtree $T'$ of
the skeleton $T$ (rooted at an arbitrary leaf), and to draw the skirted
Halin-graph $H^-(T')$.    There are three edges (called {\em connector-edges})
that connect $H^-(T')$ with the rest of $H(T)$: they attach at the root and at
the leftmost and rightmost leaf of $T'$.  To be able to add them later with a
$y$-monotone curve, we restrict the locations of their endpoints.
So we specify below whether the leftmost and rightmost
leaf should have empty rays towards west ($\W$) or east ($\E$).
We also restrict the root to be in the leftmost column and either as far north ($\N$) as
possible or as far south ($\S$) as possible; sometimes either placement is
acceptable and we use $\W$ to indicate this.  The full set of restrictions
is as as follows:

\begin{definition}  Let $T$ be a rooted tree with $\rpw(T)\geq 2$ (and therefore $\ell_L\neq \ell_R$).
Let $\Gamma$ be a plane poly-line drawing of $H^-(T)$ in layers $1,\dots,h$
(enumerated top to bottom), where $h\geq 2$.  We call $\Gamma$ an
{\em $\alpha_L \alpha_r \alpha_R$-drawing}, for $\alpha_L,\alpha_R\in \{\W,\E\}$
and $\alpha_r \in \{\N,\W,\S\}$,
if it satisfies the following (see also Figure~\ref{fig:types}):
\begin{itemize}
\item[(d1)]  $\ell_R$ is in layer $1$ and $\ell_L$ is in layer $h$.
	Root $r$ is in the leftmost column and the only element of $\Gamma$ in that column.
\item[(d2)]  For $X\in \{L,R\}$, if $\alpha_X=\W$, then the westward ray from $\ell_X$ is unobstructed
	(i.e., intersects no other element of $\Gamma$).
	Otherwise ($\alpha_X=\E$) the eastward ray from $\ell_X$ is unobstructed.
\item[(d3)]  If $\alpha_r=\N$, then $r$ is in layer 2.  If $\alpha_r=D$, then $r$ is in layer $h-1$.
	Otherwise ($\alpha_r=\W$) $r$ is in an arbitrary layer.
\end{itemize}
\end{definition}

\begin{figure}[ht]
\hspace*{\fill}
	\includegraphics[width=0.22\linewidth,page=1]{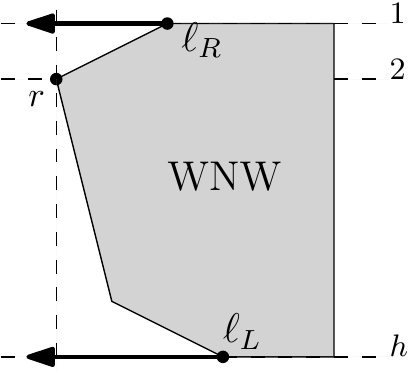}
\hspace*{\fill}
	\includegraphics[width=0.22\linewidth,page=2]{construction2.pdf}
\hspace*{\fill}
	\includegraphics[width=0.22\linewidth,page=3]{construction2.pdf}
\hspace*{\fill}
\caption{Drawing types}
\label{fig:types}
\end{figure}

We assumed $\rpw(T)\geq 2$ in the above definition since otherwise $\ell_L=\ell_R$
and then condition (d1) cannot be satisfied for $h>1$.
We hence create drawings for trees $T$ with $\rpw(T)\geq 2$
and deal with subtrees that do not satisfy this as special cases.
The construction works for both regular and extended Halin-graphs, but
the latter may require a bit more height.
To express this succinctly, set $\chi_{ext}(T)$ to be 1 if $T$ contains a
degree-2 vertex that is not the root (this in particular implies that $H(T)$
is not regular), and $\chi_{ext}(T)=0$ otherwise.  Note that $\chi_{ext}(T')\leq \chi_{ext}(T)$
for any subtree $T'$ of $T$.

\paragraph{The case $\rpw(T)=2$ and some useful observations:}

The drawing for $T$ if $\rpw(T)=2$ is a bit special; we can save two rows (compared
to drawings for higher rooted pathwidth) at the cost of no flexibility for the
$y$-coordinate of the root.

\begin{lemma}
\label{lem:main_2}
\label{lem:main2}
Let $T$ be a rooted ordered tree with $\rpw(T)= 2$.
Then for any $\alpha_L,\alpha_R\in \{\W,\E\}$
$H^-(T)$ has a plane $y$-monotone $\alpha_L \W \alpha_R$-drawing of height
$3+\chi(\alpha_L{=}\E)+\chi(\alpha_R{=}\E)+2\chi_{ext}(T)$.
\end{lemma}
\begin{proof}
See Figure~\ref{fig:baseCase} for the following construction.
Fix a spine $P$ that goes from root to a leaf, and
place $P$ on one layer, with the root leftmost.

Any $T'\in \calT(T,P)$ has rooted pathwidth 1 since $P$ is a spine.
If $\chi_{ext}(T)=0$, then $T'$ has no vertices of degree 2, so it
is a single leaf.  Place it in the layer above or below $P$ depending on whether $T'$
is right or left of the spine $P$.   The cycle-edges can now be completed
along these layers.  If $\chi_{ext}(T)=1$, then initially contract all vertices
of degree 2 and draw the tree as above.  Then insert extra layers before/after the
spine-layer and place degree-2 vertices (or a bend, if there are none) within those layers.

So we have constructed a $\LLL$-drawing of height $3+2\chi_{ext}(T)$.
Any of the other drawing-types is constructed by ``turning rays'' around.
We describe this in a more general lemma below since it will be useful
for later cases as well.
\end{proof}

\begin{figure}[ht]
\hspace*{\fill}
\begin{subfigure}[b]{0.48\linewidth}
	\includegraphics[width=0.99\linewidth,page=4]{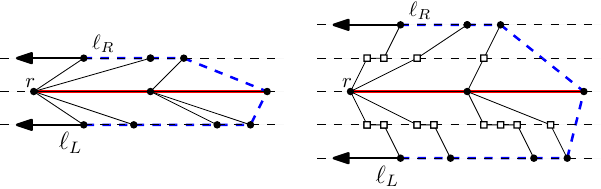}
	\caption{}
\label{fig:baseCase}
\end{subfigure}
\hspace*{\fill}
\begin{subfigure}[b]{0.45\linewidth}
	\includegraphics[width=0.99\linewidth,page=5]{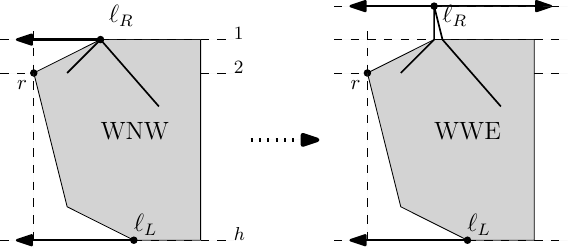}
	\caption{}
\label{fig:freeDown}
\end{subfigure}
\hspace*{\fill}
\caption{%
(a) The base case for $\chi_{ext}(T)=0$ and $\chi_{ext}(T)=1$.
(b) Achieving $\alpha_R=\E$.
}
\end{figure}

\begin{claim}
\label{cl:turnRight}
Assume that $H^-(T)$ has a $y$-monotone $\W \alpha_r \W$-drawing of height $h\geq 3$
for some $\alpha_r\in \{\N,\S,\W\}$.  Then for any $\alpha_L,\alpha_R\in \{\W,\E\}$ it also has a $y$-monotone
$\alpha_L \W \alpha_R$-drawing of height
$h+\chi(\alpha_L{=}\E)+\chi(\alpha_R{=}\E)$.
\end{claim}
\begin{proof}
Leaf $\ell_R$ is in the topmost layer, so
its incident edges are routed $y$-monotonically and leave horizontally or downward from $\ell_R$.
To achieve $\alpha_R=\E$, add a new layer above $\Gamma$, move $\ell_R$ into it, and
extend its incident edges via a bend near where $\ell_R$ used to be.
See also Figure~\ref{fig:freeDown}.
This gives a $y$-monotone drawing where the bottom layer is unchanged (in particular, it still
contains $\ell_L$ with its unobstructed ray).
The root is no longer be in $\N$-position if it was before,
but this is not a problem since we only promised an $\alpha_L \W \alpha_R$-drawing.
Similarly one achieves $\alpha_L=\E$ by adding a layer below $\Gamma$ and moving $\ell_L$ into it.
\end{proof}

The following will be useful when merging a drawing of a subtree $T'$ that
uses fewer layers than permitted (e.g. because $\chi_{ext}(T')=0$ while $\chi_{ext}(T)=1$):
We can ``pad'' such a drawing by inserting empty layers suitably, even while
maintaining the drawing type.

\begin{claim}
\label{cl:widenUU}
Assume that $H^-(T)$ has a $y$-monotone $\alpha_L \alpha_r \alpha_R$-drawing of height $h\geq 3$.
Then for any $h'>h$ it also has a $y$-monotone $\alpha_L \alpha_r \alpha_R$-drawing of height $h'$.
\end{claim}
\begin{proof}
First insert bends whenever an edge crosses a layer without a bend; now all edge-segments are
horizontal or connect adjacent layers.  If $\alpha_r=\N$
then $r$ is in layer 2.  Insert $h'-h$ horizontal grid-lines
between layer 2 and layer 3, and add bends to any edge that crosses the
inserted lines.  So edge-segments again are horizontal or connect adjacent grid-lines,
which means that we can change the $y$-coordinates of grid-lines to be integers (i.e., stretch
the drawing between layers 2 and 3) without affecting planarity or $y$-monotonicity.
This gives the desired $\alpha_L \N \alpha_R$-drawing since $r$ remains in layer 2, and
no changes were made within the top or bottom layer.
The construction is symmetric (inserting layers between $h-2$ and $h-1$) for $\alpha_r=\S$,
and either construction can be used for $\alpha_r=\W$.
\end{proof}

\paragraph{The induction hypothesis:}

We create drawings for arbitrarily large rooted pathwidth by induction; the
following states the induction hypothesis.  (It differs from Lemma~\ref{lem:main2}
in that we sometimes permit $\alpha_r=\N$ or $\alpha_r=\S$ while Lemma~\ref{lem:main2}
only holds for $\alpha_r=\W$.)

\begin{lemma}
\label{lem:main}
\label{lem:straight}
Let $T$ be a rooted ordered tree with $\rpw(T)\geq 3$, and let $\alpha_L,\alpha_r,\alpha_R$ be any of the
combinations \LLR, \RLL, \RLR, \LUL and \LDL.
Then $H^-(T)$ has a plane $y$-monotone $\alpha_L \alpha_r \alpha_R$-drawing of height
$6\rpw(T)-9+\chi(\alpha_L{=}\E)+\chi(\alpha_R{=}\E)+2\chi_{ext}(T)$.
\end{lemma}

Before proving this lemma, we briefly argue why it suffices.

\begin{theorem}
\label{thm:straight}
Every regular Halin-graph $H(T)$ has a straight-line drawing of height at most $12 \pw(T) -3$, and
every extended Halin-graph $H(T)$ has a straight-line drawing of height at most $12 \pw(T) -1$.

Specifically, the height is $6\rpw(T)-9+2\chi_{ext}(T)$ for a suitable choice of root for $T$.
\end{theorem}
\begin{proof}
Root the skeleton $T$ at a leaf $r$ such that $\rpw(T)\leq 2\pw(T)+1$
\cite{OPTII}.
Apply Lemma~\ref{lem:main2} or \ref{lem:main}
to this rooted version of $T$ to obtain a $y$-monotone \LLL-drawing of $H^-(T)$
of height $6\rpw(T)-9+2\chi_{ext}(T)\leq 12\pw(T)-3+2\chi_{ext}(T)$.
The westward ray from $\ell_L$ is unobstructed;
we can draw $(\ell_L,r)$ along this ray
until the leftmost column and then go up to $r$.
Likewise we can draw $(\ell_R,r)$
to obtain a $y$-monotone drawing of $H(T)$.
This can be transformed into a straight-line drawing
of the same height \cite{PT04,EFLN06}.
\end{proof}

The height of Theorem~\ref{thm:straight} is (roughly) a factor 2 worse than the
height in Corollary~\ref{cor:6pw}.  However, in terms of
{\em rooted} pathwidth, Theorem~\ref{thm:straight} is tight,
see Theorem~\ref{thm:lower_regular_rpw} and \ref{thm:lower_extended_rpw}.

\medskip

The rest of this section is dedicated to the proof of
Lemma~\ref{lem:main}.
It suffices to show how to construct
a \LUL-drawing of height $h:= 6\rpw(T)-9+2\chi_{ext}(T)\geq 9$;
the construction of a \LDL-drawing is symmetric and all other cases are
covered by Claim~\ref{cl:turnRight}.

We use the following notations throughout.
Let $r$ be the root of $T$, let $d$ be its degree, and let $c_1,\dots,c_d$
be the children of the root, in order.
We use the notation $\ell_L^i$ and $\ell_R^i$ (for $i=1,\dots,d$) for
the leftmost and rightmost leaf of $T_{c_i}$.
Recall that $\Sd(T)=\rpw(T)\geq 3$. Let $c_s$ be the spine-child of the root;
by definition of Horton-Strahler number this is the {\em only} child whose subtree
could have the same rooted pathwidth as $T$.
If $\rpw(T_{c_s})<\rpw(T_{c_s})$ then (to avoid some cases) we re-assign $s=d$.
Whether or not we reassigned, we hence have $\rpw(T_{c_i})<\rpw(T)$ for all $i\neq s$.

We prove the lemma by induction on $\rpw(T)$, with the base case at $\rpw(T)=3$.
We do an inner induction on the size of the tree, and use as base case the
case $\rpw(T_{c_s})<3$ (this must occur since at the leaf of the spine the rooted
pathwidth is 1).  Much of the construction will be the same for base case and induction
step, and we therefore prove them together.

\paragraph{Drawing subtrees up to $T_{c_s}$: }
The following algorithm (illustrated in Figure~\ref{fig:base})
states which drawing to use for each subtree and how to combine them.
We build the drawing
left-to-right, beginning with the root and then adding the subtrees at
the children.

\begin{figure}[ht]
	\includegraphics[width=0.48\linewidth,page=6]{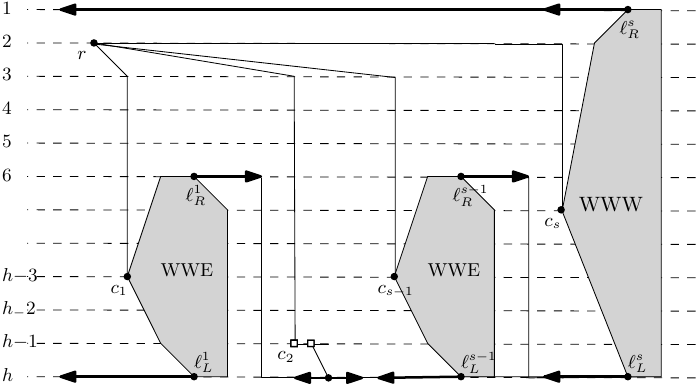}
\hspace*{\fill}
	\includegraphics[width=0.48\linewidth,page=7]{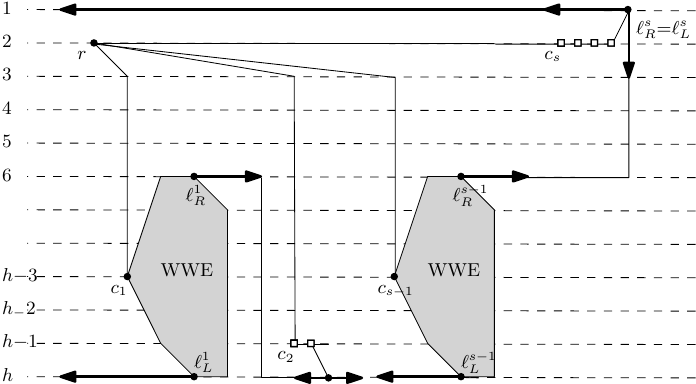}
\caption{The constructions if $c_s$ is the rightmost child.}
\label{fig:hIsLast}
\label{fig:base}
\end{figure}

\begin{enumerate}
\item Place the root $r$ in the leftmost column of layer 2.
        We {\em reserve} the eastward ray from $r$
	for edge $(r,c_d)$, i.e., we will make sure that
	nothing is added to intersect it until the edge-segment is completed.

\item \label{step:first}
	For $i=1,\dots,d-1$, if $\rpw(T_{c_i})\geq 2$, then
	recursively (or via Lemma~\ref{lem:main2}) obtain an \LLR-drawing $\Gamma_i$.
	This has height at most $6\rpw(T_{c_i})-9+1+2\chi_{ext}(T_{c_i})$.
	Since $\rpw(T_{c_i})\leq \rpw(T)-1$ and $\chi_{ext}(T_{c_i})\leq \chi_{ext}(T)$
	this is at most $6\rpw(T)-6-9+1+2\chi_{ext}(T)=h-5$.
	If needed we can use Claim~\ref{cl:widenUU}
	to make $\Gamma_i$ have height exactly $h-5$.
        Place $\Gamma_i$ in layers $6,\dots,h$, to the right
	of everything drawn thus far.

	If $\rpw(T_{c_i})=1$, then $T_i$ is a rooted path.  Place its leaf
	on layer $h$ and its degree-2 vertices (if any) on layer $h-1$, with
	the root leftmost.

	We place (parts of) the connector-edges of $T_{c_i}$ as follows:
	\begin{itemize}
	\item Connect $c_i$ to $r$ by going upward to layer 3 and then (via a bend) to layer 2.
	\item
	We draw part of the connector-edge $(\ell_R^i,\ell_L^{i+1})$
	by going eastward from $\ell_R^i$ (in its layer)
	beyond $\Gamma_i$, and adding (if needed) a bend to go downward to layer $h$.
	The eastward ray in layer $h$ from here is reserved for edge $(\ell_R^i,\ell_L^{i+1})$.
	\item
	For $i=1$, $\ell_L^i$ is the leftmost leaf of $T$; its westward ray
	is unobstructed as required.
	For $i>1$, leaf $\ell_L^i$ was placed on the ray reserved for edge
	$(\ell_R^{i-1},\ell_L^{i})$, which is hence completed.
        Since this edge receives no further bend at $\ell_L^{i}$, and was drawn
	$y$-monotonically extending from $\ell_R^{i-1}$, it is drawn $y$-monotone.
	\end{itemize}
\item \label{step:lastGood} To handle the spine-child $c_s$ we have three cases.

	Assume first that $s=d$ and $\rpw(T_{c_s})\geq 2$.  Recursively
	(or via Lemma~\ref{lem:main2}) obtain a \LLL-drawing $\Gamma_s$
	of $H^-(T_{c_s})$
	and increase its height (if needed) to be $h$.
        Place $\Gamma_s$ in layers $1,\dots,h$, to the right
	of everything drawn thus far.
	Connect $c_s$ to $r$ by going upward to layer 2 and
	then horizontally to $r$.  Edge $(\ell_R^{s-1},\ell_L^s)$
	is completed automatically, and $\ell_R^s$ is the rightmost
	leaf and its eastward ray is unobstructed.

	Assume next that $s=d$ and $\rpw(T_{c_s})=1$.  (This can happen if we
	re-assigned $s$.)  Place the leaf of $T_{c_s}$ on layer $1$ and all
	other vertices on layer 2, with the root leftmost.  (If $|T_{c_s}|=1$
	then place a bend in row $2$.)  Edge $(r,c_s)$ is
	completed automatically,
	and $\ell_R^s$ has an eastward unobstructed ray.
	To route connector-edge $(\ell_R^{s-1},\ell_L^s)$, we undo the partial
	routing that we did earlier; instead we go eastward from
	$\ell_R^{d-1}$ and then upward to $\ell_L^d=\ell_R^d$ in row 1.

	Assume finally that $s<d$, i.e., $c_s$ is not the rightmost child.
	The drawing here is much more complicated and will be explained below.
\end{enumerate}

\paragraph{Drawing the remaining subtrees if $s<d-2$: }

Our construction is done if $s=d$, so assume not.
By the re-assignment, this implies $\rpw(T_{c_s})=\rpw(T)$.
In particular, we are not in the base case of the inner induction, and
we know that $\rpw(T_{c_s})\geq 3$.  This allows us (crucially) to
choose an \LDL-drawing for $H^-(T_{c_s})$, which in turn permits us
to route $(r,c_s)$ $y$-monotonically while leaving sufficiently
much space for $T_{c_{s+1}},\dots,T_{c_d}$.

We assume for now that $s\leq d-2$; the case $s=d$ has been dealt with above
and the case $s=d-1$ is not difficult but requires a variation that will
be explained below.

\begin{figure}[ht]
\hspace*{\fill}
	\includegraphics[width=\linewidth,page=8]{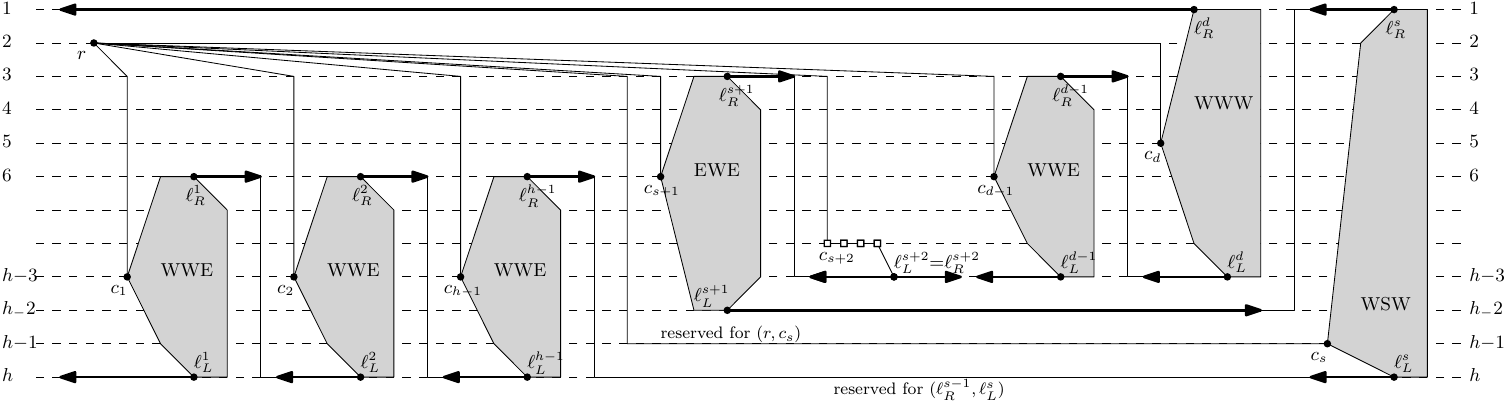} \\[2ex]
	\includegraphics[width=\linewidth,page=9]{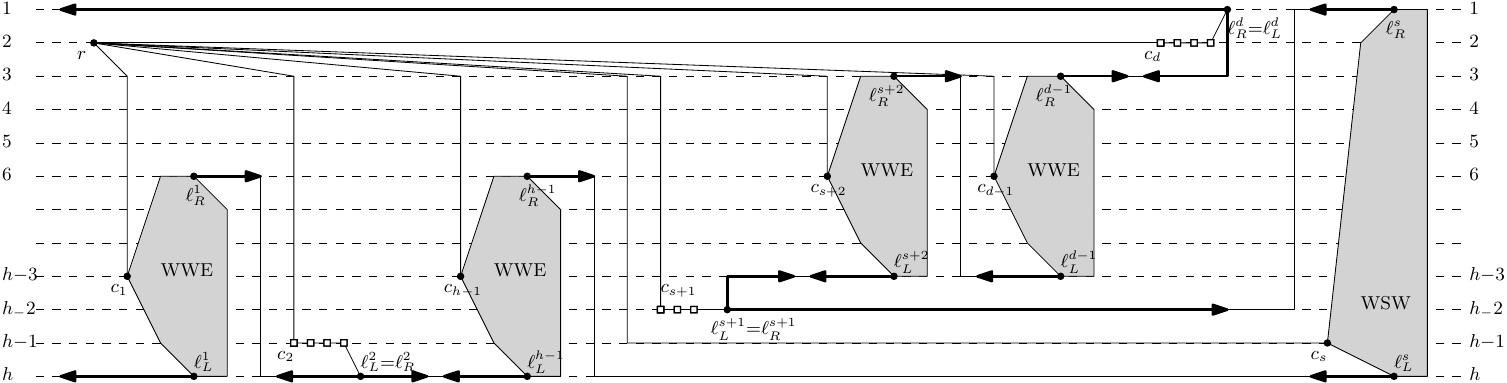}
\hspace*{\fill}
\caption{(Top) The construction if $s\leq d-2$.
(Bottom) The construction again, with other subtrees as rooted paths.}
\label{fig:construction}
\label{fig:rpw1normal}
\label{fig:rpw1hplus1}
\label{fig:rpw1d}
\end{figure}

The drawing algorithm continues as follows if $s\leq d-2$:
\begin{enumerate}
\addtocounter{enumi}{3}
\item Draw parts of the edge $(r,c_s)$, by going from $r$ to a bend
	in layer 3 (to the right of everything drawn thus far), then down to another bend in layer
	$h-1$.  We reserve the eastward ray in layer $h-1$ from this bend for edge $(r,c_s)$.
\item \label{step:RLR}
	By $s\leq d-2$ child $c_{s+1}$ exists and is not $c_d$.

	If $\rpw(T_{c_{s+1}})\geq 2$, then let $\Gamma_{s+1}$ be a recursively obtained \RLR-drawing of $T_{c_{s+1}}$;
	since $\rpw(T_{c_{s+1}})< \rpw(T)$ this has height at most $h-4$.  Increase its
	height (if needed) so that it has height exactly $h-4$, and place $\Gamma_{s+1}$ in layers $3,\dots,h-2$ to the right
	of everything drawn thus far.    If $T_{c_{s+1}}$ is a rooted path, then place all its vertices in layer $h-2$.
	Draw the connector-edges as follows:

	\begin{itemize}
	\item Connect $c_{s+1}$ to $r$ by going upward to layer 3 and then (via a bend) to layer 2.
	\item Leaf $\ell_{L}^{s+1}$ is placed in layer $h-2$; we reserve its eastward ray for edge
$(\ell_L^{s+1},\ell_R^{s})$.
	\item If $T_{c_{s+1}}$ is not a rooted path, then leaf $\ell_{R}^{s+1}$ has an unobstructed eastward ray; we begin drawing
	edge $(\ell_R^{s+1},\ell_L^{s+2})$ by going eastward
	from $\ell_R^{s+1}$, then vertically to layer $h-3$ and reserving the eastward ray
	in layer $h-3$ for $(\ell_R^{s+1},\ell_L^{s+2})$.

	If $T_{c_{s+1}}$ is a rooted path, then $\ell_R^{s+1}$ is in layer $h-2$.  We go up one unit to layer $h-3$
	and reserve the eastward ray for  $(\ell_L^{s+1},\ell_R^{s})$.
	\end{itemize}

\item For $i=s+2,s+3,\dots,d-1$, we process $T_{c_i}$ and its connector-edges
	as we did in Step~\ref{step:first}, only we put the drawing three levels higher.
\item \label{step:LULd}
	We process $T_{c_d}$ very similarly to Step~\ref{step:lastGood}.

	So assume first that $\rpw(T_{c_d})\geq 2$.  Recursively
        (or via Lemma~\ref{lem:main2}) obtain an \LLL-drawing $\Gamma_d$
	of $H^-(T_{c_d})$ of height at most $h-6$.  Increase its height to be $h-3$.
        Place $\Gamma_d$ in layers $1,\dots,h-3$, to the right
	of everything drawn thus far.
	Connect $c_d$ to $r$ by going upward to layer 2 and
	then horizontally to $r$.  Edge $(\ell_R^{d-1},\ell_L^d)$
	is completed automatically, and $\ell_R^d$ is the rightmost
	leaf and its eastward ray is unobstructed.

	Now assume that $T_{c_d}$ is a rooted path.
	Place the leaf of $T_{c_d}$ on layer $1$ and all
	other vertices on layer 2, with the root leftmost (if $|T_{c_s}|=1$,
	then place a bend in row $2$).  Edge $(r,c_s)$ is
	completed automatically, and $\ell_R^d$ (which
	is the rightmost leaf) has an eastward unobstructed ray.
	To route connector-edge $(\ell_R^{d-1},\ell_L^d)$, we have two
	cases.  If $d>s+2$ and/or $\rpw(T_{c_{d-1}})\geq 2$, then undo the
	partial routing that we did earlier; instead we go eastward from
	$\ell_R^{d-1}$ and then upward to $\ell_L^d=\ell_R^d$ in row 1.
	If $d=s+2$ and $\rpw(T_{c_{d-1}})=1$, then the partial drawing
	of $(\ell_R^{d-1},\ell_L^d)$ is unusual since
	$\ell_R^{d-1}=\ell_R^{s+1}$ was placed in layer $h-2$, and the
	edge was routed by going upward to layer $h-3$.   We now go eastward
	from there and then upward to layer 1. This is the only situation where
	a connector-edge receives bends when placing both endpoints, but one
	verifies that this route is $y$-monotone.

\item Recall from Step~\ref{step:RLR} that the eastward ray in layer $h-2$ was reserved for connector-edge $(\ell_L^{s+1},\ell_R^s)$.
	We now add a bend in it to the right of everything drawn thus far, then go vertically to
	layer 1 and reserve the eastward ray.

\item Finally, recursively obtain an \LDL-drawing $\Gamma_s$ of $T_{c_s}$ of
	height $h$.
	(This exists since $c_s$ is not the rightmost child, hence
	$\rpw(T_{c_s})=\rpw(T)\geq 3$ and induction can be applied.)
	Place $\Gamma_s$ to the right of everything drawn thus far.
	Since $\ell_R^s,c_s,\ell_L^s$ are in layers $1,h-1$ and $h$, respectively, this completes
	the connector-edges of $T_{c_s}$.
\end{enumerate}

\paragraph{The case $s=d-1$:}
	Previously, we used an
	\RLR-drawing for $c_{s+1}$ and an \LUL-drawing for $c_d$ in Steps~\ref{step:RLR} and \ref{step:LULd}.  If $s=d-1$,
	then $c_{s+1}=c_d$ takes on the roles of both of these drawings.  The following step
(see Figure~\ref{fig:rpw1normal})
replaces step 5 and 6 if $s=d-1$.

\begin{figure}[ht]
	\includegraphics[width=0.48\linewidth,page=11]{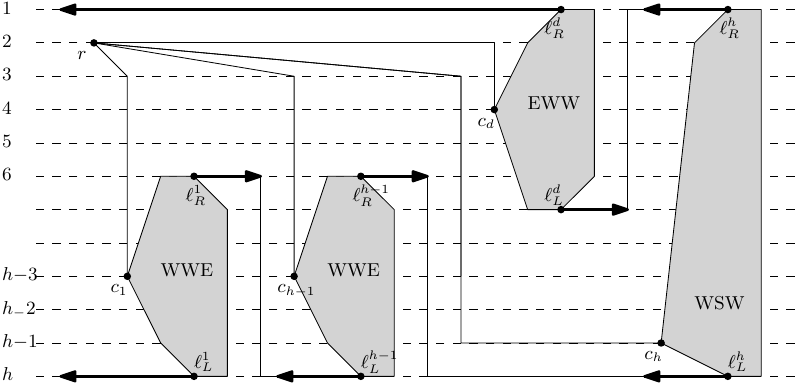}
\hspace*{\fill}
	\includegraphics[width=0.48\linewidth,page=10]{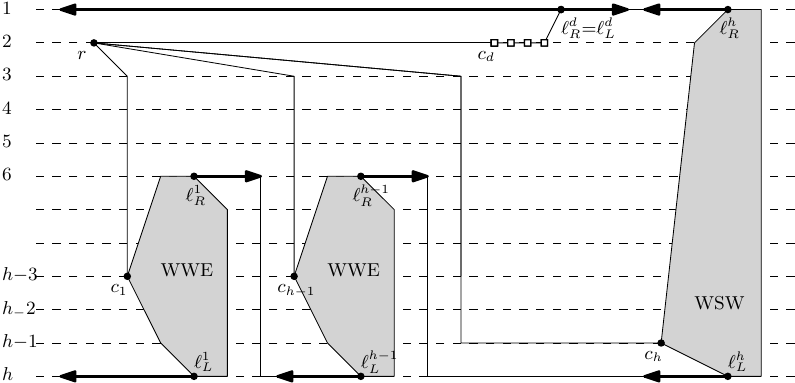}
\caption{The constructions if $s=d-1$.}
\label{fig:hIsdMinus1}
\end{figure}

\begin{itemize}
%
%
\item[5'.]

	If $\rpw(T_d)\geq 2$, then let $\Gamma_d'$ be an \RLL-drawing of
	$T_{c_d}$ of height $h-5$, and place it in layers $1,\dots,h-5$.  All
	connector-edges can be completed as before, with the only change
	that $(\ell_L^d,\ell_R^{s})$ now uses layer $h-5$ rather than $h-2$.
	If $\rpw(T_d)=1$, then place $\ell_L^d=\ell_R^d$ in layer 1 and
	the rest of $T_d$ in layer 2.
	See Figure~\ref{fig:hIsdMinus1}.
\end{itemize}

We have constructed a \LUL-drawing in all cases,
and one easily verifies
that all edges are drawn $y$-monotonically,
hence Lemma~\ref{lem:straight} and with it Theorem~\ref{thm:straight} holds.

It is worth mentioning that this poly-line drawing can easily
be found in linear time, as long as coordinates of vertices are
expressed initially with via offsets to their parents,
and evaluated to their final value only after finishing the
construction of the entire tree.

\subsection{Halin-graphs with maximum degree 3}

Observe that in Figures~\ref{fig:hIsLast} and \ref{fig:hIsdMinus1} (where $s\in \{d-1,d\}$) we are ``wasting'' layers; the same construction could have been done with three fewer layers.  This leads to the following.

\begin{lemma}
  \label{lem:3rpw}
Let $T$ be a rooted \emph{binary} tree with
$\rpw(T)\geq 2$, and let $\alpha_L,\alpha_r,\alpha_R$ be any of the
combinations \LLR, \RLL, \RLR, \LUL and \LDL.
Then $H^-(T)$ has a plane $y$-monotone $\alpha_L \alpha_r \alpha_R$-drawing of height
$3\rpw(T)-3+\chi(\alpha_L{=}\E)+\chi(\alpha_R{=}\E)+\chi_{ext}(T)$.
\end{lemma}
\begin{proof}
We again proceed by induction and show that there exists
an \LUL-drawing of height $3\rpw(T)-2$ (all other drawing-types
are symmetric or obtained with Claim~\ref{cl:turnRight}).
We only sketch the necessary changes to the previous algorithm
here; the reader should be able to fill in the details using Figure~\ref{fig:maxDeg3}.

\begin{figure}[ht]
	\includegraphics[width=0.2\linewidth,page=13]{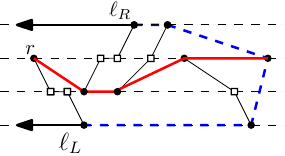}
\hspace*{\fill}
	\includegraphics[width=0.3\linewidth,page=12]{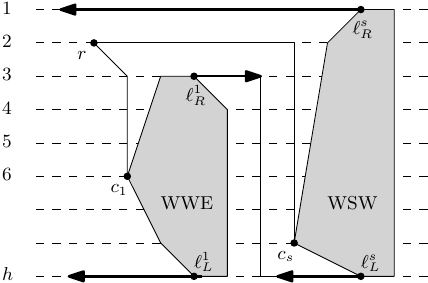}
\hspace*{\fill}
	\includegraphics[width=0.3\linewidth,page=14]{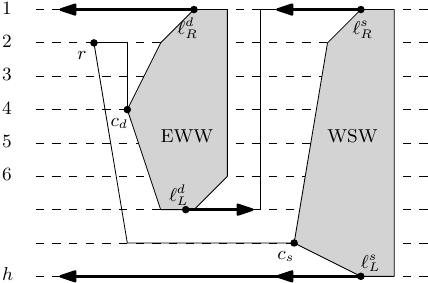}
\caption{The constructions if the maximum degree is 3.}
\label{fig:maxDeg3}
\end{figure}

The previous base case construction gives
$3+2\chi_{ext}(T)$ layers.
We can also achieve at most 4 layers, by placing a spine-vertex on
layer 2 if the spine-child is the right child and on layer 3 otherwise.
Using the better of the two (depending on $\chi_{ext}(T)$) we hence have
$3+\chi_{ext}(T)=3\rpw(T)-3+\chi_{ext}(T)$ layers.
In the induction step, we have $d\leq 2$ children and
hence always either $s=d$ or $s=d-1$.
So construct a drawing of $H^-(T)$ as in Figure~\ref{fig:hIsLast} or Figure~\ref{fig:hIsdMinus1},
except use $h=3\rpw(T)-3+\chi_{ext}(T)$ and place drawing $\Gamma_i$ for $i<s$ in layers $3,\dots,h$.
\end{proof}

If an extended Halin-graph has maximum degree 3, then its skeleton $T$ is binary when rooting it at a leaf.
Since we can do so and achieve $\rpw(T)\leq 2\pw(T)+1$, this implies
as in the proof of Theorem~\ref{thm:straight}:

\begin{theorem}
Every extended Halin-graph with maximum degree 3 and
skeleton $T$ has a straight-line drawing of
height at most $6\pw(T) +\chi_{ext}(T)$.
\end{theorem}


%% file: lower2.tex
Both papers that gave approximation algorithms for the height on
tree drawings \cite{Sud04,BB-JGAA} also constructed trees
where this bound is tight.    In particular, Batzill and Biedl
showed that there exists an ordered tree that requires height
$2\pw(T)+1$ in any ordered drawing \cite{BB-JGAA}.
In the same spirit, we now construct
Halin-graphs that need as much height as we achieve
with our algorithms.
\footnote{The graphs were chosen as to keep the argument
as simple as possible; like much smaller trees would do.}

\begin{definition}
\label{def:lower}
For $w\geq 1$, define $C_w$ and $F_w$ as follows:
\begin{itemize}
\item 
	$C_1$ consists of a path $\langle r,c\rangle$ (where $r$
	is the root) with a leaf attached at each of them on each side of the path.
	See Figure~\ref{fig:lowerBoundC2reg}.
\item $F_w$ is obtained from $C_w$ as follows.  Let $r$
	be the root of $C_w$.  Add a parent $p$ and a grand-parent $g$ to $r$ and make $g$
	the root.  Attach
	a leaf on each side of path $\langle p,g,r\rangle$ at each of $p,g$.
	See Figure~\ref{fig:lowerBoundFw}.
\item $C_{w+1}$ is obtained as follows.  Start with a {\em spine}
	consisting of vertices $s_1,\dots,s_S$ for some sufficiently large constant $S$
	that we will specify later, and make $s_1$ the root.
	At each spine-vertex except $s_S$,
	attach on each side of the spine $L$ copies of $F_w$ via its root, for some sufficiently
	large constant $L$ that we will specify later.
	See Figure~\ref{fig:lowerBoundCw}.
\end{itemize}
\end{definition}

\begin{figure}[ht]
\hspace*{\fill}
\begin{subfigure}[b]{0.15\linewidth}
\includegraphics[width=0.99\linewidth,page=7]{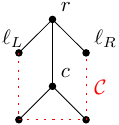}
\caption{}
\label{fig:lowerBoundC2reg}
\end{subfigure}
\hspace*{\fill}
\begin{subfigure}[b]{0.18\linewidth}
\includegraphics[width=0.99\linewidth,page=8]{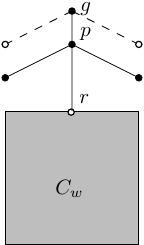}
\caption{}
\label{fig:lowerBoundFw}
\end{subfigure}
\hspace*{\fill}
\begin{subfigure}[b]{0.18\linewidth}
\includegraphics[width=0.99\linewidth,page=9]{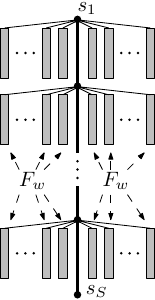}
\caption{}
\label{fig:lowerBoundCw}
\end{subfigure}
\hspace*{\fill}
\begin{subfigure}[b]{0.18\linewidth}
\includegraphics[width=0.99\linewidth,page=3]{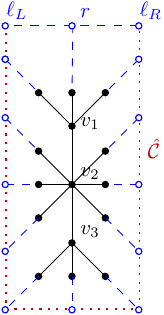}
\caption{}
\label{fig:lowerBoundC2ext}
\end{subfigure}
\hspace*{\fill}
\caption{A Halin-graph requiring much more height than its pathwidth. (a) Tree $C_1$
with cycle $\calC$ (cycle-edges are dotted red) that encloses $c$.
(b) Obtaining $F_w$ from $C_w$.  Dashes edges are not needed except to avoid
degree-2 vertices in the trees.  (c) Obtaining $C_{w+1}$
using many copies of $F_w$.  (d) Tree $\hat{C}_2$ needed for Theorem~\ref{thm:lower_extended_rpw}.}
\label{fig:lowerBound}
\label{fig:lowerBound1}
\end{figure}

Define $h(w)$ to be 3 if $w=1$ and $h(w):=h(w{-}1)+6 = 6w-3$ otherwise.
The main ingredient for our lower bound is the following result:

\begin{lemma}
\label{lem:lower}
For $w\geq 1$, any plane poly-line drawing of $H^-(C_w)$ uses at least $h(w)$ layers.
\end{lemma}

We prove Lemma~\ref{lem:lower} by induction on $w$.
In the base case ($w=1$) vertex $c$ in $C_1$ is surrounded by
a 5-cycle in $H^-(C_1)$.  Since we need one layer for $c$, and two more layers
to surround it, any plane drawing of $H^-(C_1)$ requires three layers as desired.
The induction step will proved over the next four subsections,
but we sketch here the main idea.  Fix an arbitrary plane poly-line drawing $\Gamma$
of $H^-(C_{w+1})$ for some $w\geq 1$.  Tree $C_{w+1}$
contains lots of copies of $F_w$, hence of $C_w$. Therefore,
$\Gamma$ contains lots of copies of $H^-(C_w)$; each of them uses at least $h(w)$ layers
by induction.  We can argue
that some copy of $H^-(F_w)$ inside $\Gamma$ actually requires $h(w)+1$ layers;
this is the most difficult part that we defer to last.
Furthermore, there are 5 polylines inside $\Gamma$ that are disjoint from
this copy of $H^-(F_w)$ and that ``bypass'' it (defined below).
It is known that 5 bypassing polylines
need 5 additional layers.  Therefore the height is at least $h(w)+1+5\geq h(w+1)$.

\subsection{Preliminaries and preprocessing}

We first introduce some terms concerning the abstract tree $C_{w+1}$.
Recall that $C_{w+1}$ is rooted and has a total order among the children of every vertex.
We therefore have a total order among the leaves, starting at the leftmost leaf
and ending with the rightmost one.  However, we will use
``left'' to refer to the order of vertices within one level of the drawing,
which may or may not reflect the order in the tree.  To avoid confusion, we will therefore treat
the order of chidren/leaves as if it were time, and so speak of the ``first''/``last'' leaf
and that a leaf comes earlier than another.

We distinguish leaves of $C_{w+1}$
(other than $s_S$) by whether they are on the {\em before-spine} or {\em after-spine}, i.e., before
or after $s_S$ in the enumeration of leaves.  Likewise for a spine-vertex $s_i\neq s_S$
we distinguish the non-spine children by whether they are before or after the spine.
Any such non-spine child $g$ is the root of a copy of $F_w$ which we
denote by $F(g)$.  For any two leaves $\ell,\ell'$ of $C_{w+1}$,
the {\em cycle-path from $\ell$ to $\ell'$} consists of the subpath of the cycle-edges between
$\ell$ and $\ell'$.

\medskip
Now we introduce some terms concerning drawing $\Gamma$.
Enumerate the layers of $\Gamma$, from top to bottom,
as $1,2,\dots,h$.  We are done if $h\geq h(w+1)$, so assume
for contradiction that $h<h(w+1)=h(w)+6$.
In fact we may assume $h=h(w)+5$ because we can add empty layers.
For two points $p,q$, we write $p\prec q$ (or ``$p$ is {\em left} of $q$'')
if $p$ and $q$ are on the same layer and $p$ has smaller $x$-coordinate.

A few minor modifications to drawing $\Gamma$ will make later arguments
easier and do not affect the height.    First, insert a bend into any
edge-segment that crosses a layer without having a bend there.
(These new bends may not have integral $x$-coordinates, but integrality
of $x$-coordinate is never used in the lower-bound proof.)
Second, do the following for any spine-vertex $s_i$ (with $ i < S $) of $C_{w+1}$, and
any non-spine child $g$ of $s_i$.  Recall that $g$ had three children; one is
vertex $p$ while two are leaves.  Delete the two edges to these leaves;
their sole purpose was to ensure that the Halin-graph is regular and they
will not be used in the proof.   With this, $g$ now has degree 2.
For the third modification, if $(s_i,g)$ is not drawn as a straight-line, then move $g$ to the bend  on $(s_i,g)$
nearest to $s_i$. This makes $(s_i,g)$ a straight-line and (by the first
preprocessing step) puts $g$ either on the same level as $s_i$ or one level above or below; this will be frequently used below.

Recall that $F(g)$ denotes the copy of $F_w$ attached at $g$.  We use $\Gamma(g)$ for the
drawing of $H^-(F(g))$ as it appears after these modifications.  Since $\Gamma(g)$ contains a
drawing of $H^-(C_w)$ within, it must use at least $h(w)$ layers.

\medskip
Finally we briefly review
the concept of bypassing (see also Figure~\ref{fig:bypass}); we use a version
here that is 90$^\circ$ rotated from the one in \cite{OPTII}.
Recall that bends of a polyline (like all bends and vertices of $\Gamma$) are
required to have integral $y$-coordinates.

\begin{definition}
Consider a set of poly-lines $\pi_1,\dots,\pi_k$ that are disjoint except
(perhaps) at their endpoints.
Let $\hat{\pi}$ be a poly-line that is disjoint from $\pi_1,\dots,\pi_k$.
We say that $\pi_1,\dots,\pi_k$ {\em bypass} $\hat{\pi}$ if there
exists a layer $\ell$  that intersects $\hat{\pi}$,   and for $i=1,\dots,k$
poly-line $\pi_i$ begins and ends in layer $\ell$ and
all points in $\hat{\pi}\cap \ell$ are between the two ends of $\pi_i$.
\end{definition}

\begin{lemma}\cite{OPTII}
\label{lem:bypass}
If a planar poly-line drawing $\Gamma$
contains $k$ poly-lines that bypass a poly-line $\hat{\pi}$,
and if $\hat{\pi}$ intersects $h$ layers, then $\Gamma$ uses
at least $h+k$ layers.
\end{lemma}

\subsection{The ideal case}

We first argue that the height-bound holds in one special case;
we will show later that this situation must
occur somewhere in $C_{w+1}$ (up to symmetry), as long as $S$ and $L$ are big enough.
We assume that the following holds
(see also Figure~\ref{fig:GoodCaseSetup}):

\begin{figure}[th]
\begin{subfigure}[b]{0.2\linewidth}
\includegraphics[scale=0.5,page=4,trim=50 0 30 40,clip]{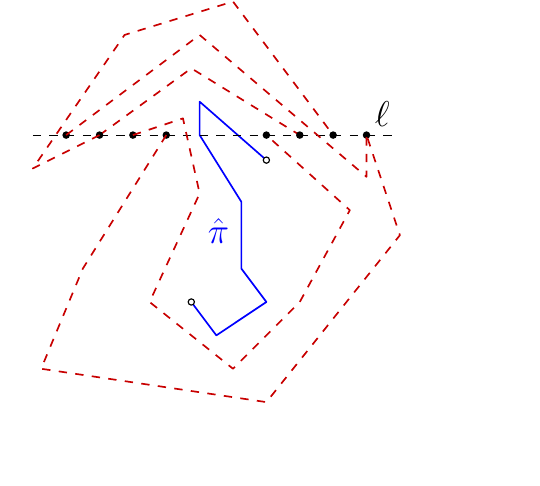}
\caption{}
\label{fig:bypass}
\end{subfigure}
\hspace*{\fill}
\begin{subfigure}[b]{0.4\linewidth}
\includegraphics[scale=0.77,page=2,trim=10 0 50 30,clip]{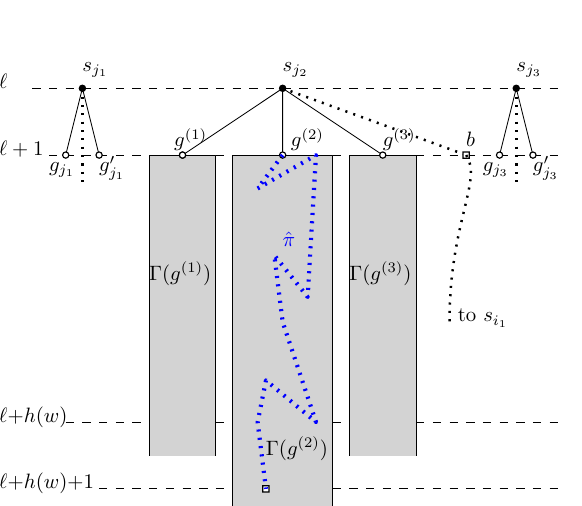}
\caption{}
\label{fig:GoodCaseSetup}
\label{fig:pocket}
\end{subfigure}
\hspace*{\fill}
\begin{subfigure}[b]{0.3\linewidth}
\includegraphics[scale=0.77,page=5]{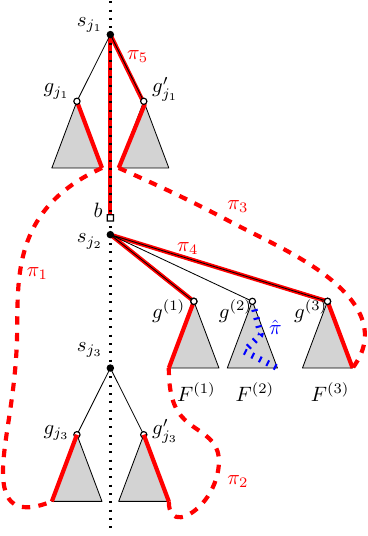}
\caption{}
\label{fig:fivePaths}
\end{subfigure}
\hspace*{\fill}
\caption{(a) Five bypassing paths. (b) The ideal case. Path
$\hat{\pi}$ is blue (thick dotted).  Spine-edges are thin dotted.
(c) Five bypassing paths in $H^-(C_{w+1})$.  Cycle-edges are dashed.}
\end{figure}

\begin{itemize}
\itemsep -1pt
\item[(C1)]  There are three spine-vertices $s_{j_1},s_{j_2},s_{j_3}$
	that are all located in one layer $\ell\leq 5$.
	Furthermore, $1\leq j_1<j_2<j_3<S$ and $s_{j_1}\prec s_{j_2}\prec s_{j_3}$.
\item[(C2)] For $k=1,2,3$, vertex $s_{j_k}$ has an after-spine child $g_{j_k}$ and a before-spine child $g_{j_k}'$ on layer $\ell{+}1$.
	In fact, $s_{j_2}$ has five after-spine children on layer $\ell{+}1$.
\item[(C3)]
Vertex $s_{j_2}$ has three after-spine children $g^{(1)},g^{(2)},g^{(3)}$ on layer $\ell{+}1$ for which
	$g^{(1)}\prec g^{(2)}\prec g^{(3)}$.
	The order of children at $s_{j_2}$ contains $g^{(1)},g^{(2)},g^{(3)}$ as subsequence.

Furthermore, one of the spine-edges incident to $s_{j_2}$ has a bend or endpoint $b$ on layer $\ell+1$.
	If $b$ is on edge $(s_{j_2},s_{j_2-1})$ then $g^{(3)}\prec b$,
	otherwise $b\prec g^{(1)}$.
\item[(C4)]
For $k=1,2,3$, drawing $\Gamma(g^{(k)})$ occupies
		no point on layer $\ell$ or above.
\end{itemize}

We will later argue that the following property holds automatically, given (C1-C4).

\begin{itemize}
\item[(C5)] There exists a path $\hat{\pi}$ within $\Gamma(g^{(2)})$ that connects $g^{(2)}$ (which is on layer $\ell+1$)
	to layer $\ell+h(w)+1$, and all points in $\hat{\pi}\cap (\ell+1)$ lie strictly between $g^{(1)}$ and $g^{(3)}$.
\end{itemize}

Now we define five interior-disjoint paths in $C_{w+1}$ as follows:
(see also Figure~\ref{fig:fivePaths}):

\begin{itemize}
\item $\pi_1$:  This path begins at $g_{j_1}'$, continues within $F(g_{j_1}')$ to the last leaf, and
	from there along the cycle-path to the first leaf of $F(g_{j_3}')$.  From there it goes upwards in the tree to $g_{j_3}'$.
	This path uses only $F(g_{j_1}')$ and $F(g_{j_3}')$ and cycle-edges among
	leaves that are before the spine.

\item $\pi_2$: This path begins at $g_{j_3}$, continues within $F(g_{j_3})$ to the last leaf, and
	from there along the cycle-path to the first leaf of $F(g^{(1)})$.
	From there it goes upwards in the tree to $g^{(1)}$.
	This path uses only $F(g_{j_3})$ and $F(g^{(1)})$ and cycle-edges among
	leaves that are between $s_S$ and the first leaf of $F(g^{(1)})$ in the total order of leaves.

\item $\pi_3$: This is simply the path $\langle g^{(1)},s_{j_2},g^{(3)}\rangle$, which uses only edges incident
	to $s_{j_2}$.

\item $\pi_4$: This path is built symmetrically to $\pi_2$: begin at $g_{j_1}$, go to the first leaf of $F(g_{j_1})$,
	from there along the cycle-path (in reverse) to the last leaf of $F(g^{(3)})$, and from there to $g^{(3)}$.
	This path uses only $F(g_{j_1})$ and $F(g^{(3)})$ and cycle-edges among
	leaves that the last leaf of $F(g^{(3)})$ or later.

\item $\pi_5$:    Recall that one bend $b$ of a spine-edge incident to $s_{j_2}$ lies on layer $\ell+1$.
	Path $\pi_5$ begins at $b$, and goes along spine-edges, away from $s_{j_2}$, until it reaches
	either $s_{j_1}$ or $s_{j_3}$.  From there it goes to the after-spine child on layer $\ell+1$, i.e., either
	$g_{j_1}$ or $g_{j_3}$.  Except for this last edge, $\pi_5$ uses only spine-edges.
\end{itemize}

\begin{claim}
The polylines corresponding to
paths $\pi_1,\dots,\pi_5$ bypass $\hat{\pi}$.
\end{claim}
\begin{proof}
Directly from the edges that they use, one observes that the five paths are disjoint from $\hat{\pi}$, and from each other except that they may have endpoints in common.  (We use here that $g^{(2)}$ lies between $g^{(1)}$ and $g^{(3)}$ in the order of children at $s_{j_2}$ by (C3).)
Assume that $b$ is right of $g^{(3)}$, the other case is symmetric.  Then all five paths begin at a point in
$\{g_{j_1},g_{j_1}',g^{(1)}\}$ and end at a point in
$\{g^{(3)},b,g_{j_3},g_{j_3}'\}$.
Observe that $g_{j_1}$
is necessarily left of $g^{(1)}$, otherwise the straight-line segments $(s_{j_1},g_{j_1})$ and
$(s_{j_2},g^{(1)})$ would intersect.  Likewise $g_{j_1}'\prec g^{(1)}$ and $g^{(3)}\prec g_{j_3}, g_{j_3}'$.
So all five paths connect a point on layer $\ell+1$ that is at or to the left of $g^{(1)}$ with a point on layer $\ell+1$ that is at or to the right of $g^{(3)}$.
Since $\hat{\pi}$ uses only points on $(\ell+1)$ that are strictly between $g^{(1)}$ and $g^{(3)}$ by (C5), the claim holds.
\end{proof}

Since $\hat{\pi}$ spans $h(w)+1$ layers,
therefore drawing $\Gamma$ of $H^-(C_{w+1})$ has at
least $(h(w)+1)+5=h(w{+}1)$ layers as desired.

\subsection{Guaranteeing conditions (C1-C4)}

Now we argue that conditions (C1-C4) are satisfied at some
subtrees if $S$ and $L$ are big enough.  Recall that we assumed
(for contradiction) that $h=h(w)+5$.  Since each copy of $H^-(F_w)$
uses at least $h(w)$ layers, we therefore have only 5 layers for
bypassing any copy of $H^-(F_w)$.    Roughly speaking, this forces
spine-vertices to be in the top 5 or the bottom 5 layers.
Therefore (C1) holds if $S$ is big enough.  Next we argue that of
the $L$ attached copies of $F_w$ at a spine-vertex $s$, only
$L-72$ can share a layer with $s$.  This, plus the preprocessing,
forces (C2) if $L\geq 81$.  It also implies that
many non-spine children satisfy (C4), and an appropriate choice among them
ensures (C3).

To give the details, we first study various properties of non-spine
children of one fixed spine-vertex $s_i$ with $i<S$.

\begin{observation}
\label{obs:intersectMiddle}
For any non-spine child $g$ of $s_i$, $\Gamma(g)$ intersects all layers in $\{6,\dots,h(w)\}$.
\end{observation}
\begin{proof}
There are $h= h(w)+5$ layers in total, and
by induction $\Gamma(g)$ intersects at least $h(w)$ layers.
It therefore can avoid only the top 5 and the bottom 5 layers.
\end{proof}

We say that $g$ is {\em bad} if the layer of $s_i$ intersects $\Gamma(g)$,
otherwise $g$ is {\em good}.

\begin{claim}
\label{cl:outsideRange}
At most 72 non-spine children of $s_i$ are bad.%
\footnote{With more effort one can show that at most 12 of them can be bad,
leading to a better bound for $L$.}
\end{claim}
\begin{proof}
We say that a non-spine child $g$ has {\em type} $(t,b)$ if the
topmost and bottommost layer used by $\Gamma(g)$ are $t$ and $b$.
By Observation~\ref{obs:intersectMiddle} we have
$1\leq t\leq 6$ and $h(w)\leq b\leq h(w)+5$, so
there are at most 36 types.  Assume for contradiction that there are
$73=2\cdot 36+1$ bad non-spine children of $s_i$, hence
three of them (say $g_1,g_2,g_3$) have the same type $(t,b)$.

For $k=1,2,3$, let $B_k$ be a poly-line within $\Gamma(g_k)$ that begins in layer $t$ and
ends in layer $b$.
Let $Q_k$ be a poly-line that starts at $s_i$ (which is within layers $\{t,\dots,b\}$ since $g_k$ is bad), goes along the straight-line edge to $g_k$ (also within $\{t,\dots,b\}$) and continues within $\Gamma(g_k)$ until it
reaches $B_k$.  Note that $B_1\cup Q_1$ and $B_2\cup Q_2$ and $B_3\cup Q_3$ are disjoint except at $s_i$,
and reside entirely within layers $\{t,\dots,b\}$.
See also Figure~\ref{fig:K33}.

Exactly as in the proof of Lemma 5 in \cite{BB-JGAA}, one argues that this is impossible.
Consider the drawing induced by $\bigcup_k (B_k\cup Q_k)$.  Add a vertex $v'$ in layer $t-1$
and connect it to the top ends of $B_1,B_2,B_3$ (they are in layer $t$).  Likewise add a
vertex $v''$ in layer $b+1$ and connect it to the bottom ends of $B_1,B_2,B_3$ (they are
in layer $b$).  This gives a planar drawing of $K_{3,3}$, with $\{s_i,v',v''\}$ as one side
and the points $B_k\cap Q_k$ for $k=1,2,3$ as the other side.
Contradiction.
\end{proof}

\begin{figure}[ht]
\hspace*{\fill}
\begin{subfigure}[b]{0.34\linewidth}
\includegraphics[width=0.99\linewidth,page=1]{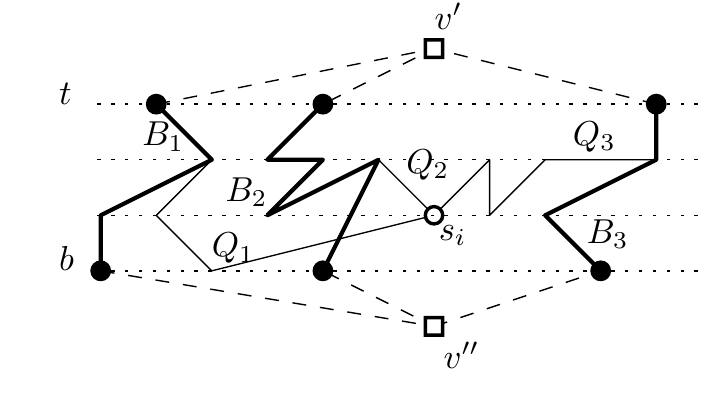}
\caption{}
\label{fig:K33}
\end{subfigure}
\hspace*{\fill}
\begin{subfigure}[b]{0.6\linewidth}
\includegraphics[page=6]{lowerBound2.pdf}
\caption{}
\label{fig:cyclicShiftsA}
\end{subfigure}
\hspace*{\fill}
\caption{%
(a) Three bad non-spine children of type $(t,b)$ imply a planar drawing of $K_{3,3}$.
(Picture based on \cite{BB-JGAA}).
(b) Possible arrangements of non-spine children of $s_{j_2}$ on layer $\ell+1$.
}
\end{figure}

\begin{corollary}
\label{cor:spineLayer}
If $L\geq 37$ then the layer of $s_i$ is in $\{1,\dots,5\} \cup \{h(w){+}1,\dots,h(w){+}5\}$.%
\footnote{With more effort one can show that $s_i$ cannot be on the topmost or bottommost
layer for $i>1$, leading to a better bound for $S$.}
\end{corollary}
\begin{proof}
If $s_i$ were in any layer in $\{6,\dots,h(w)\}$, then by Observation~\ref{obs:intersectMiddle}
all $2L\geq 74$ non-spine children of $s_i$ would be bad.
\end{proof}

\begin{claim}
\label{cl:manyGood}
If $s_i$ is on layer $\ell$ where $\ell \leq 5$ and $\ell\leq h/2$,
and if $L\geq 81$, then $s_i$ has at least $5$ good after-spine children
on layer $\ell+1$.
\end{claim}
\begin{proof}
There are $L$ after-spine children, hence at least $L-72\geq 9$ that are good.
Any such good child $g$ cannot be on layer $\ell$ by definition of good, and
it is at most one layer away by the preprocessing.  So $g$ is on layer
$\ell-1$ or $\ell+1$.  Assume for contradiction that there at most
4 good after-spine children on layer $\ell+1$.
So at least 5 good after-spine children are on layer $\ell-1$,  call them
$g_1,\dots,g_5$, enumerated in left-to-right order along the layer.

We now have two cases.  In the first case, $\ell\leq h(w)$ (which is
always true for $w\geq 2$ since then $h(w)\geq 9$ while $\ell\leq 5$).
Since $g_1$ is good, drawing $\Gamma(g_1)$ cannot use layer $\ell$, so it is contained
within layer $1,\dots,\ell-1$.  So it uses at most $h(w)-1$
layers, which is impossible.

In the second case, $\ell>h(w)$.  Then $w=1$, hence $h(w)=3$, so
$\ell>3$.  But we also know that $h=h(w)+5=8$  and $\ell\leq h/2=4$, so $\ell=4$.
Let $\Gamma'$ be the drawing obtained by flipping $\Gamma$ upside down.
Since there were 8 layers, $s_i$ is now located on layer $\ell'=5$, children
$g_1,\dots,g_5$ are on layer 6, and their drawings only use
layers 6,7,8.

Since edge $(s_i,g_k)$ (for $k=1,\dots,5$) is drawn straight-line
by the pre-processing, and $\Gamma$ respects the planar embedding,
the cyclic order of neighbours of $s_i$ must contain $g_1,\dots,g_5$
\todo{perhaps add figure}
in this order.  The spine-edges and before-spine children at $s_i$
may appear somewhere between $g_1$ and $g_5$ in the cyclic order,
but regardless of where they are, either $g_1,g_2,g_3$ or
$g_3,g_4,g_5$ are a subsequence of the linear order of children of $s_i$.
By Claim~\ref{cl:oneTallChild} (proved below, but there
is no circularity) drawing $\Gamma(g_2)$ or $\Gamma(g_4)$ hence uses a point on layer
$\ell'+h(w)+1= 5+3+1=9$. This gives the required contradiction of our assumption.
\end{proof}

Now we explain how to satisfy (C1)-(C4).
Assuming $S\geq 42$, we have 41 spine-vertices $s_i$ with $i<S$.
Assuming $L\geq 37$, each of them is on one of 10 possible layers by Corollary~\ref{cor:spineLayer}.
By the pigeon-hole principle, therefore, at least 5 of these spine-vertices are on one layer $\ell$.
After a possible vertical flip of $\Gamma$,
we may assume $\ell\leq h/2$, therefore $\ell\leq 5$ by Corollary~\ref{cor:spineLayer}.%
\footnote{Note that flipping the drawing reverses all edge-orders, so we might be proving
a lower bound for $C_{w+1}^{rev}$, the tree $C_{w+1}$ with all orders of children reversed.
But $C_{w+1}^{rev}$ is isomorphic to $C_{w+1}$, so their skirted graphs are isomorphic and this is not a problem.}
Among the 5 spine-vertices on $\ell$, we can (by the Erd\H{o}s-Szekeres theorem~\cite{Erdos1935})
find a subsequence of $\lceil \sqrt{5} \rceil = 3$ spine-vertices $s_{j_1},s_{j_2},s_{j_1}$
such that $j_1<j_2<j_3$ and either
$s_{j_1}\prec s_{j_2} \prec s_{j_3}$ or $s_{j_3}\prec s_{j_2}\prec s_{j_1}$.
After a possible horizontal flip of $\Gamma$ we have $s_{j_1}\prec s_{j_2}\prec s_{j_3}$
and therefore (C1) holds.

(C2) holds (assuming $L\geq 81$) due to Claim~\ref{cl:manyGood} and a symmetric lemma, proved
exactly the same way, for before-spine children.

To argue (C3), let $g_1,\dots,g_5$ be the 5 after-spine
children of $s_{j_2}$ that are good and on layer $\ell+1$, enumerated in left-to-right order
along the layer.
Let $g'$ be a before-spine child of $s_{j_2}$ that is on layer $\ell+1$, and notice
that the cyclic order of neighbours of $s_{j_2}$ contains
$\langle g',s_{j_2+1},g_1,\dots,g_5,s_{j_2-1}\rangle=:\rho$
as subsequence.  Since the edges from $s_{j_2}$ to $g',g_1,\dots,g_5$ are
straight-line by the pre-processing, the $x$-coordinate order of $g',g_1,\dots,g_5$ along layer $\ell+1$ must fit
the (cyclic) order $\rho$.
Depending on whether $g_3$ is right or left of $g'$, therefore
either $g'\prec g_1\prec g_2 \prec g_3$ or $g_3 \prec g_4 \prec g_5\prec g'$,
See Figure~\ref{fig:cyclicShiftsA}.

If $g'\prec g_1 \prec g_2\prec g_3$ then spine-edge $(s_{j_2},s_{j_2+1})$ leaves $s_{j_2}$
between the two segments $(s_{j_2},g')$ and $(s_{j_2},g_1)$; this forces
the spine-edge to go to layer $\ell+1$ as well, and by the pre-processing
it either ends there or it receives a bend $b$ there with $b\prec g_1$.  So (C3) holds
for $\{g^{(1)},g^{(2)},g^{(3)}\}:=\{g_1,g_2,g_3\}$.
Similarly if $g_3\prec g_4\prec g_5\prec g'$ then
spine-edge $(s_{j_2},s_{j_2-1})$ ends or receives a bend $b$ on
layer $\ell+1$ with $g_5\prec b$, and (C3) holds
for $\{g^{(1)},g^{(2)},g^{(3)}\}:=\{g_3,g_4,g_5\}$.

Finally (C4) holds since the chosen vertices $\{g^{(1)},g^{(2)},g^{(3)}\}$ were
good and on layer $\ell+1$ and so drawings $\Gamma(g^{(1)}),\Gamma(g^{(2)},\Gamma(g^{(3)})$
cannot use layer $\ell$ or above.

\subsection{Arguing (C5)}
So we have now found subtrees such that (C1-C4) hold.  This always implies (C5), but
the argument for this is lengthy.   We also need to prove the missing piece for
Claim~\ref{cl:manyGood}.  Both will be done with the same argument as follows.

\begin{claim}
\label{cl:oneTallChild}
Let $s_i$ (for $i<S$) be a spine-vertex on layer $\ell$
that has three good after-spine children $g^{(1)},g^{(2)},g^{(3)}$
on layer $\ell+1$ and the order of children at $s_i$ contains $g^{(1)},g^{(2)},g^{(3)}$ as subsequence.
Then there exists a path $\hat{\pi}$ within $\Gamma(g^{(2)})$ that connects $g^{(2)}$
to layer $\ell+h(w)+1$, and all points in $\hat{\pi}\cap (\ell+1)$ lie between $g^{(1)}$ and $g^{(3)}$.
\end{claim}
\begin{proof}
Recall that tree $F_w$ is built by extending
tree $C_w$; let $C$ be the copy of $C_w$ that is inside $F(g^{(2)})$.  Also let
$\calI$ be the open interval of points on layer $\ell+1$ between $g^{(1)}$ and $g^{(3)}$,
so path $\pi$ should intersects layer $\ell+1$ only in $\calI$.  We need an observation.

\begin{observation}
$H^-(C)$ uses no points in $\calI$.
\end{observation}
\begin{proof}
Define a cycle $Q$ in $H^-(C_{w+1})$ as follows.  Start at the unique child $p$
of $g^{(2)}$, go to its last child $\ell_R$ (which is a leaf)
and from there along the cycle-path to the first leaf of $F(g^{(3)})$.  Go
upwards in tree $F(g^{(3)})$ to $g^{(3)}$ and from there to $s_i$.  Continue
symmetrically through $F^{(1)}$, i.e., go from $s_i$ to $g^{(1)}$ to
the last leaf of $F(g^{(1)})$, then along the cycle-path to the first
child $\ell_L$ of $p$ and then to $p$.
See Figure~\ref{fig:pocketQ}.  This cycle separates $g^{(2)}$ from $H^-(C)$
in the planar embedding since $g^{(2)}$ is between $g^{(1)}$ and $g^{(3)}$ in
the order of children of $s_i$.

Now study the corresponding poly-line $Q$ in $\Gamma$.
Since $\langle g^{(1)},
s_i,g^{(3)}\rangle$ is drawn with straight-line segments between layers $\ell+1$
and $\ell$, and since $g^{(2)}\in \calI$ and $\Gamma$ is plane, all of $\calI$ is on or inside $Q$.  On the other hand $H^-(C)$ is strictly outside $Q$ and the claim holds.
\end{proof}

\begin{figure}[ht]
\hspace*{\fill}
\begin{subfigure}[b]{0.4\linewidth}
\includegraphics[width=0.99\linewidth,page=4]{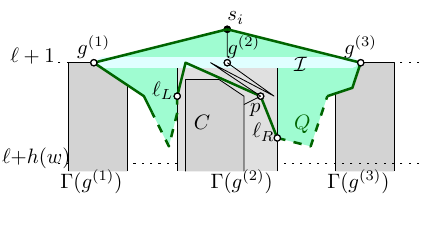}
\caption{}
\label{fig:pocketQ}
\end{subfigure}
\begin{subfigure}[b]{0.4\linewidth}
\includegraphics[width=0.99\linewidth,page=5]{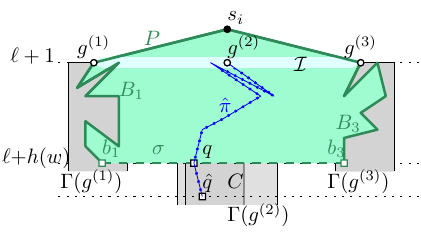}
\caption{}
\label{fig:pocketP}
\end{subfigure}
\hspace*{\fill}
\caption{For the proof of Claim~\ref{cl:oneTallChild}.
(a) Poly-line $Q$ separates $\calI$ from $C$.
(b) The pocket $P$.}
\end{figure}

Let the {\em pocket} $P$ be defined as follows, see also Figure~\ref{fig:pocketP}.
For $k=1,3$, let $B_k$
be a poly-line within $\Gamma(g^{(k)})$ that connects $g^{(k)}$ to a point $b_k$ on layer $\ell+h(w)$;
this exists since $\Gamma(g^{(k)})$ spans at least $h(w)$ layers and contains no point in layer $\ell$.
We choose $b_k$ such that $B_k$ is minimal, i.e., contains no other point on layer $\ell+h(w)$; in
particular all its points are hence in layers $\ell+1,\dots,\ell+h(w)$.
Let the {\em lid} $\sigma$ be the line-segment $\overline{b_1b_3}$; note that $\sigma$ is not necessarily a
segment of $\Gamma$.  Now define pocket $P$ to be the set bounded by
$B_1\cup \langle g^{(1)},s_i,g^{(3)}\rangle \cup B_3 \cup \sigma$,
where the lid $\sigma$ is included in $P$ while all other points on the boundary are excluded.
Note that any point in $(\ell+1)\cap P$ is in $\calI$, because $B_1$ and $B_3$
contain no points on layer $\ell$ or above by (C4).

Assume for contradiction that all of $\Gamma(g^{(2)})$ (and in particular therefore
$H^-(C)$) resides within pocket $P$.  Then $H^-(C)$ uses no points
on layer $\ell+1$, because it does not use points in $\calI$.
Therefore $H^-(C)$ fits within $h(w)-1$ layers, a contradiction.
So $\Gamma(g^{(2)})$ must use points outside the pocket.   These cannot
be on $B_1\cup B_3$ or $\langle g^{(1)},s_i,g^{(3)} \rangle$  since these
paths do not belong to $F(g^{(2)})$.   So to get to a point outside $P$,
some polyline of $\Gamma(g^{(2)})$ must contain
a point $q$ on $\sigma\subset P$ from which it goes downward.
Let $\hat{q}$ be the next bend of this polyline, which is
on layer $\ell+h(w)+1$ by the preprocessing.  Let $\hat{\pi}$ be the
poly-line from $g^{(2)}$ (on layer $\ell+1$) to point $\hat{q}$ (on layer
$\ell+h(w)+1$) that is within $\Gamma(g^{(2)})$.  With the
exception of the segment from $q$ to $\hat{q}$, poly-line $\hat{\pi}$
was inside pocket $P$; in particular it can use no points on layer $\ell+1$
except the ones that are on $\calI$.  This proves the claim.
\end{proof}

So we have proved Claim~\ref{cl:oneTallChild}, which finishes the proof
of Claim~\ref{cl:manyGood}.  Hence (C2) holds. From this we derived
(C3) and (C4), hence the precondition for Claim~\ref{cl:oneTallChild}
holds for the three children $g^{(1)},g^{(2)},g^{(3)}$ of $s_{j_2}$ that we chose.
Claim~\ref{cl:oneTallChild}
hence implies (C5) and the proof of Lemma~\ref{lem:lower} is complete.

\subsection{Proving the lower bounds}

We now finally prove the lower bounds.  To do so, we first bound the (rooted)
pathwidth of $F_w$ and trees derived from it.

\begin{observation}
We have $\rpw(F_w)\leq w+1$ and $\pw(F_w'')\leq w-1$, where $F_w''$ is the leaf-reduced inner skeleton of $H(F_w)$.
\end{observation}
\begin{proof}
We proceed by induction on $w$.  Tree $F_1$ consists of a path $\langle g,p,r,c\rangle$
with leaves attached; this has rooted pathwidth 2.
Also $F_1''$ consists only of $g$, since it is obtained from $F_1$ by first deleting
all leaves (this gives a path), and then repeatedly doing leaf-reductions
(this removes all but $g$).  So $\pw(F_1'')=0$.

Now consider $F_{w+1}$ for $w\geq 1$.  This consists of a path $\langle g,p,s_1,\dots,s_S\rangle$
with copies of $F_{w}$ attached.  Using this path as spine, we immediately get
$\rpw(F_{w+1})\leq \rpw(F_w)+1 \leq w+2$.  Also, $F_{w+1}''$ consists of the same path
with copies of $F_w''$ attached; therefore $\pw(F_{w+1}'')\leq \pw(F_w'')+1\leq w$.
\end{proof}

Thus far all constructions and lower bounds have been for {\em plane} drawings (respecting the embedding and have the cycle-edges at the infinite region). But we can easily prove lower bounds even for planar drawings
which have no requirement except to be crossing-free.

\begin{theorem}
\label{thm:lower_regular_pw}
There exists a regular Halin-graph $H(T)$ such that
any planar poly-line drawing of $H(T)$ requires at least $6\pw(T'')+3$ layers,
where $T''$ is the reduced tree of the inner skeleton of $H(T)$.
\end{theorem}
\begin{proof}
For any $w\geq 2$, consider the tree $T$ obtained by taking two copies
of $F_w$ and combining them by adding an edge between the two copies of
the root $g$.
Fix an arbitrary planar poly-line drawing $\Gamma$ of $H(T)$.
Since $H(T)$ is 3-connected
\cite{Halin71} the clockwise order of edges must be the same in $H(T)$
and in $\Gamma$.  But the infinite region of $\Gamma$ could be incident to
some face different from the one bounded by the cycle-edges.  Tree $T$ contains two copies of $F_w$, and the infinite region of $\Gamma$ can be a
face of $H^-(F_w)$ for at most one of them.    Therefore $\Gamma$ contains
a plane drawing of $H^-(F_w)$, hence also one of $H^-(C_w)$.
By Lemma~\ref{lem:lower} this requires at least $h(w)=6w-3$ layers.
The reduced inner skeleton of $H(T)$ consists of two copies of $F_w''$,
each of which had pathwidth at most $w-1$, and this bound is obtained
with a main path that ends at $g$.  Therefore we can use the two combined
paths as main path for $T''$ and so have $\pw(T'')\leq w-1$
and the bound holds.
\end{proof}

We note that this lower bound implies a lower bound of $\Omega(\log n)$ on the height,
since $C_w$ contains $c^w$ vertices for some (rather large) constant $c$.  However,
this bound is not new since already using the Halin-graph of a complete ternary tree
could give a lower bound of $\Omega(\log n)$ on the height.
The main contribution of our lower bound is that it matches the
upper bound relative to ``$\pw(T'')$'' in Theorem~\ref{thm:3x-transform}.
(This was also the reason why we used the leaf-reduced inner skeleton,
rather than the skeleton, in Theorem~\ref{thm:3x-transform}.)

We also promised a lower bound in terms of the rooted pathwidth.  Note that the skeleton
of a Halin-graph is an unrooted tree $T$; to be able to talk about $\rpw(T)$
we define this to be the minimum over all choices of the root.

\begin{theorem}
\label{thm:lower_regular_rpw}
There exists a regular Halin-graph $H(T)$ such that
any planar poly-line drawing of $H(T)$ requires at least $6\rpw(T)-9$ layers.
\end{theorem}
\begin{proof}
For any $w\geq 2$, again let $T$ be two copies of $F_w$, combined by
adding an edge between the two roots.  We know $\rpw(F_w)\leq w+1$, and
the same holds for $T$ if we root it suitably.  Namely, the spine of
$F_w$ is $g$-$p$-$s_1$-$\dots$-$s_S$; if we root $T$ at one copy of $s_S$
then we can use as its spine the two combined spines of the two copies of
$F_w$ and have the same rooted pathwidth.
$H(T)$ is a regular Halin-graph and since (as above) any planar drawing of it
includes a plane drawing of $H^-(C_w)$, by Lemma~\ref{lem:lower}
it requires at least $h(w)=6w-3 \geq 6\rpw(T)-9$ layers.
\end{proof}

Because these lower bounds hold for regular Halin-graphs, they also hold
for extended Halin-graphs, but we can improve
the lower bound of Theorem~\ref{thm:lower_regular_rpw} ever so slightly
for extended Halin-graphs (hence make it tight).

\begin{theorem}
\label{thm:lower_extended_rpw}
There exists an extended Halin-graph $H(T)$ such that
any planar poly-line drawing of $H(T)$ requires at least $6\rpw(T)-7$ layers.
\end{theorem}
\begin{proof}
We give the lower bound only for a plane poly-line drawing; it can be converted
to one for planar poly-line drawings by doubling the tree as above.

We construct a rooted tree $\hat{C}_w$ that differs from $C_w$ only in the base case.
See Figure~\ref{fig:lowerBoundC2ext}.
Start with the tree $T_1$ from \cite{BB-JGAA} that
requires 3 layers in any order-preserving plane drawing.
This tree consists of a path $\langle v_1,v_2,v_3\rangle$, with three leaves attached
at each of $v_1,v_3$, and six leaves attached at $v_2$, three on each side
of the path.  To obtain $\hat{C}_2$, attach a degree-1 vertex at
every degree-1 vertex of $T_1$, and let $r$ be the middle of the new degree-1
vertices near $v_1$.  Make $r$ the root, and add
two further leaves $\ell_L,\ell_R$ that are children of $r$ and
become leftmost and rightmost leaf of the resulting tree $\hat{C}_2$.
Note that $H^-(\hat{C}_2)$ consists of a cycle $\hat{\calC}$ (using
the cycle-edges and the path $\langle \ell_L,r,\ell_R\rangle$) that
surrounds $T_1$.  Any plane poly-line drawing of $H^-(\hat{C}_2)$ therefore requires
5 layers because $\hat{\calC}$ encloses the drawing of $T_1$ that uses 3 layers.
Also note that $\rpw(\hat{C}_2)=2$.

Now construct $\hat{F}_w$ from $\hat{C}_w$ and $\hat{C}_{w+1}$ from $\hat{F}_w$ exactly as
done in Definition~\ref{def:lower}.  Set $\hat{h}(2)=5$ and $\hat{h}(w)=\hat{h}(w-1)+6$ for $w\geq 3$.
Then $H^-(\hat{C}_w)$ requires $\hat{h}(w)$ layers in any plane poly-line drawing, because this holds
for $\hat{C}_2$, and is proved for $\hat{C}_w$ for $w\geq 3$ exactly as the induction
step of Lemma~\ref{lem:lower}.    Also as before $\rpw(\hat{C}_w)\leq \rpw(\hat{C}_{w-1})+1$
for $w\geq 3$, therefore $\rpw(\hat{C}_w)\leq w$.
So any plane drawing of $H(\hat{C}_w)$ (which includes $H^-(\hat{C}_w)$) must use
$\hat{h}(w)=6w-7 \geq 6\rpw(\hat{C}_w)-7$ layers.
\end{proof}